\documentclass[11pt,conference,onecolumn]{IEEEtran}
\usepackage{amssymb}
\usepackage{amsthm}
\usepackage{amsmath}
\usepackage{amsfonts}
\usepackage{graphicx}
\usepackage{graphics}
\usepackage{subfigure}
\usepackage{multirow}
\usepackage{setspace}
\usepackage{anysize}
\marginsize{0.5in}{0.3in}{0.3in}{1in}

\newcommand{\limi}[1]{\;\dot{#1}\;}
\newcommand{\outage}[1]{P_{\rm out}(\rho,#1)}

\newtheorem{thm}{\it{Theorem}}

\newtheorem{lem}{Lemma}

\title{Time-Out Lattice Sequential Decoding for the MIMO ARQ Channel}

\author{
\authorblockN{Walid Abediseid and Mohamed Oussama Damen}
\authorblockA{Department of Elect. \& Comp. Engineering\\
University of Waterloo \\
Waterloo, Ontario N2L 3G1 \\
wabedise@engmail.uwaterloo.ca, modamen@ece.uwaterloo.ca}
}
\begin{document}
\maketitle
\begin{abstract}
The optimal diversity-multiplexing-delay tradeoff for the multi-input multi-output (MIMO) automatic repeat
request (ARQ) channel can be achieved using incremental redundancy lattice space-time codes coupled with a list decoder for joint error detection and correction. Such a decoder is based on the minimum mean-square error lattice decoding principle which is implemented using sphere decoding algorithms. However, sphere decoders suffer from high computational complexity for low-to-moderate signal-to-noise ratios, especially for large signal dimensions. In this paper, we would like to construct a more efficient decoder that is capable of achieving the optimal tradeoff with much lower complexity. In particular, we will study
the throughput-performance-complexity tradeoffs in sequential decoding algorithms and the
effect of preprocessing and termination strategies. We show, analytically and via simulation, that using the \textit{lattice sequential decoder} that implements a time-out algorithm for joint error detection and correction, the optimal tradeoff of the MIMO ARQ channel can be achieved with significant reduction in decoding complexity.
\end{abstract}
\section{Introduction}
Automatic Repeat reQuest (ARQ) is an efficient communication strategy that uses feedback to achieve high reliability, and is widely used in many wireless networks (e.g., LTE and WiMAX) (refer to~\cite{0} for a detailed study about several ARQ schemes). In its early stages, ARQ was used in conjunction  with codes with good error detection capabilities. However, such codes increase the number of retransmissions which significantly reduce transmission rate (throughput) and increase delay. This may become undesirable for many communication systems, particularly in wireless fading channels. To overcome such problems, hybrid-ARQ system was introduced which uses forward error correction techniques (e.g., block and convolutional codes) \cite{1}--\cite{4}. This, however, comes at the expense of increasing the complexity of the receiver. The design of low complexity receivers for ARQ systems that achieve \textit{near} optimal performance and high throughput is considered a challenging problem.

The class of sequential decoders is among the most promising decoders that can handle high data rates with low decoding complexity. Sequential decoders \cite{Wozen}, \cite{Label5} are tree search algorithms that were originally constructed to decode convolutional codes transmitted via discrete memoryless channel, and are well-known to achieve near optimal performance with very low (average) decoding complexity. However, there is still a non-zero probability that the decoding complexity (time) becomes excessive, especially when the channel is very noisy. In this situation, the decoder encounters buffer overflow that results in a decoding failure. It is this probability that limits the performance of the sequential decoder. Fortunately, the decoding failure probability can be totally eliminated using systems with feedback channel. For that reason, sequential decoders were adopted with ARQ systems \cite{2}--\cite{Ort} due to their ability to detect for retransmission before ending the decoding search which results in huge saving in decoding complexity while maintaining high throughput. All of this makes sequential decoding very promising and attractive for use in systems with repeat request.

Many sequential decoding algorithms (e.g., stack algorithm \cite{Label5}) were modified for the use of signal detection and decoding in ARQ systems. Among those algorithms that is considered simple but efficient is the so-called \textit{time-out} sequential decoding. In this algorithm, the decoder simply tracks the number of computations performed by the decoder and asks for retransmission if the computations become excessive and exceed a certain predetermined time limit. This results in reducing the decoding complexity by terminating the search during high channel noise. For the case of single-input single-output discrete memoryless channel, it was shown (see \cite{2}) that there exists an optimal time-out limit value that maximizes both performance and throughput while achieving low decoding complexity. In this paper, we would like to extend the work in \cite{2} to the multi-input multi-output (MIMO) ARQ channel, particularly, to the quasi-static, Rayleigh-fading MIMO ARQ channel. In particular, we will study the throughput-performance-complexity tradeoffs in sequential decoding algorithms and the effect of preprocessing and termination strategies such as the time-out algorithm.

Diversity-multiplexing trade-off (a result that has been developed by Zheng and Tse \cite{Label8} which shows a rigorous fundamental tradeoff between the data rate increase possible via \textit{multiplexing} versus the channel error probability reduction possible via \textit{diversity}) has become a standard tool to evaluate the asymptotic performance of the MIMO channels at the high signal-to-noise ratio (SNR) regime. Space-time coding, a powerful coding technique, is used to exploit such tradeoff \cite{Tarokh2}--\cite{Label2}. In MIMO ARQ systems, the delay introduced by the channel provides a third dimension in the tradeoff region. To achieve the optimal \textit{diversity-multiplexing-delay} tradeoff of the MIMO ARQ channel, El Gamal \textit{et. al.} \cite{Label1} proposed an efficient coding scheme called Incremental Redundancy LAttice Space-Time (IR-LAST) coding. These lattice-based construction of space-time codes were designed using linear random coding techniques \cite{Label2}. The problem of constructing explicit optimal IR-LAST codes for the above mentioned MIMO ARQ channel was discussed in~\cite{Label3}. However, in both papers, a list lattice decoder for joint error detection and correction (usually implemented via sphere decoding algorithms \cite{Label10}) is an essential part for achieving the optimal tradeoff. The draw back of using such a decoder is that at low-to-moderate SNR and for large signal dimensions, the size of the candidate list could become extensively large. This motivates us to search for a more efficient joint decoding technique that is capable of achieving the optimal tradeoff with a fairly low decoding complexity.

In this paper, we propose a different approach for joint detection and decoding based on the efficient \textit{lattice} stack sequential decoder~\cite{Label5}. Such a decoder was developed in~\cite{Label7} for solving the closest lattice point search problem, which is related to the optimum decoding rule in MIMO channels. We implement the time-out algorithm at the decoder to predict in advance the occurrence of high channel noise. This results in less wasted time trying to decode a noisy signal and hence improving upon decoding complexity.  We show, analytically and via simulation, that the optimal tradeoff can be achieved using such a decoder with significant reduction in (average) decoding complexity.

The organization of this paper is as follows. In Section II we introduce our system model and review basic concepts of lattices and the IR-LAST coding scheme. In Section III, we briefly describe the mechanism of lattice sequential decoder and derive lower and upper bounds for the performance of such decoder from a lattice point of view. These bounds are considered the primary elements in constructing our new MIMO ARQ decoder. In section IV, we introduce the time-out algorithm and prove its optimality in terms of the diversity-multiplexing-delay tradeoff. Finally, we verify our theoretical results using simulation in section V.

Throughout the paper, we use the following notation. The superscript $^c$ denotes complex quantities, $^\mathsf{T}$ denotes transpose, and $^\mathsf{H}$ denotes Hermitian transpose. We refer to $g(z)\limi{=}z^a$ as $\lim_{z\rightarrow\infty}g(z)/\log(z)=a$, $\dot{\geq}$ and $\dot{\leq}$ are used similarly. For a bounded Jordan-measurable region $\mathcal{R}\subset\mathbb{R}^m$, $V(\mathcal{R})$ denotes the volume of $\mathcal{R}$. We denote $\mathcal{S}_m(r)$ by the $m$-dimensional hypersphere of radius $r$ with $V(\mathcal{S}_m(r))=(\pi r^2)^{m/2}/\Gamma(m/2+1)$, where $\Gamma(x)$ denotes the Gamma function. Also, $\pmb{I}_m$ denotes the $m\times m$ identity matrix, and $\otimes$ denotes the Kronecker product. The complement of a set $\mathcal{A}$ is denoted by $\overline{\mathcal{A}}$.

\section{System Model}
\subsection{ARQ MIMO Channel}
Consider a MIMO ARQ system with $M$-transmit and $N$-receive antennas, a maximum of $L$ rounds, no channel state information (CSI) at the transmitter and perfect CSI at the receiver. For the MIMO ARQ channel model, we follow in the footsteps of El Gamal \textit{et al.} \cite{Label1} and use the incremental-redundancy ARQ transmission scheme. We restrict ourselves to the one-bit feedback (ACK/NACK) MIMO ARQ model. The ARQ feedback channel is assumed to be zero-delay and error-free. The complex baseband model of the received signal at the $\ell$-th round can be mathematically described as
\begin{equation}
\pmb{Y}^c_{\ell}=\sqrt{\rho\over M}\pmb{H}^c_{\ell}\pmb{X}^c_{\ell}+\pmb{W}^c_{\ell},
\label{channel}
\end{equation}
where $\pmb{X}^c_{\ell}\in\mathbb{C}^{M\times T}$ is the transmitted signal matrix, $T$ is the number of channel uses, $\pmb{Y}^c_{\ell}\in\mathbb{C}^{N\times T}$ is the received signal matrix, $\pmb{W}^c_{\ell}\in\mathbb{C}^{N\times T}$ is the noise matrix, $\pmb{H}^c_{\ell}\in\mathbb{C}^{N\times M}$ is the channel matrix, and $\rho$ is the average signal-to-noise ratio (SNR) per receive antenna. The elements of both the noise matrix and the channel fading gain matrix are assumed to be independent identically distributed zero-mean circularly symmetric complex Gaussian random variables with variance $\sigma^2=1$.

In this paper, we assume two different scenarios of channel dynamics. The first model being the \textit{long-term static} channel, where the channel coefficients remain constant during all $L$ rounds, i.e., $\pmb{H}^c_{\ell}=\pmb{H}^c$ for all $1\leq\ell\leq L$. The second scenario is the \textit{short-term static} channel, where the channel remains constant during each round and changes independently at each round. Also, the following short-term average power constraint on the transmitted signal is assumed
\begin{equation}\label{power}
\mathbb{E}\{\|\pmb{X}^c_{\ell}\|_F^2\}\leq MT.
\end{equation}

The equivalent real-valued channel model, after $\ell$ transmission rounds, corresponding to (\ref{channel}) can be written as
\begin{equation}\label{real model}
\pmb{y}_{\ell}=\pmb{H}_{\ell}\pmb{x}+\pmb{w}_{\ell},
\end{equation}
where we define
$\pmb{x}=(\pmb{x}^\mathsf{T}_{1,1},\dots,\pmb{x}^\mathsf{T}_{L,1},\dots,\pmb{x}^\mathsf{T}_{L,T})^\mathsf{T},$
with $\pmb{x}_{\ell,t}^\mathsf{T}=(\Re\{[\pmb{X}^c_{\ell}]_t\}^\mathsf{T},\Im\{[\pmb{X}^c_{\ell}]_t\}^\mathsf{T})^\mathsf{T} $, and
$\pmb{w}=(\pmb{w}^\mathsf{T}_{1,1},\dots, \pmb{w}^\mathsf{T}_{\ell,1},$ $\dots,\pmb{w}^\mathsf{T}_{\ell,T})^\mathsf{T},$
with $\pmb{w}_{\ell,t}^\mathsf{T}=(\Re\{[\pmb{W}^c_{\ell}]_t\}^\mathsf{T},\Im\{[\pmb{W}^c_{\ell}]_t\}^\mathsf{T})^\mathsf{T} $. The vector $\pmb{y}_{\ell}\in\mathbb{R}^{2NT\ell}$ represents the total signal received over all transmitted blocks from 1 to $\ell$. The equivalent real-valued channel matrix, $\pmb{H}_{\ell}$, has dimension $2NT\ell\times 2MTL$ and is formed by taking the first $2NT\ell$ rows of the matrix $\pmb{H}_L$ which is composed by $L$ diagonal blocks, where each diagonal block takes the form
$$\sqrt{\rho\over M}\pmb{I}_{T}\otimes\begin{pmatrix}
    \Re\{\pmb{H}^c_{\ell}\}  & -\Im\{\pmb{H}^c_{\ell}\}\\
   \Im\{\pmb{H}^c_{\ell}\}  &  \Re\{\pmb{H}^c_{\ell}\}
\end{pmatrix}.$$

There are several interesting issues for why we resort to the real-valued channel model. First, the input-output relation describing the channel that is given in (\ref{real model}) allows for the use of \textit{lattice theory}~\cite{Conway} which may simplify our analysis. Second, space-time codes based on lattices have been widely used in MIMO channels due to their low encoding complexity (e.g., nested or Voronoi codes) and the capability of achieving excellent error performance~\cite{Label2}. Another important aspect of lattice space-time (LAST) codes is that they can be decoded by a class of efficient decoders known as  \textit{lattice decoders} (implemented using sphere decoding algorithms).

\subsection{IR-LAST Coding Scheme}
An $m$-dimensional lattice code $\mathcal{C}(\Lambda,\pmb{u}_o,\mathcal{R})$ is the finite subset of the lattice translate $\Lambda+\pmb{u}_0$ inside the shaping region $\mathcal{R}$, i.e., $\mathcal{C}=\{\Lambda+\pmb{u}_0\}\cap\mathcal{R}$, where $\mathcal{R}$ is a bounded measurable region of $\mathbb{R}^m$. It is well-known \cite{Label2} that an $(M\times T)\times L$ space-time coding scheme is a full-dimensional LAST code if its vectorized (real) codebook (corresponding to the channel model (\ref{real model})) is a lattice code with dimension $m=2MTL$.  Let $\Lambda_c=\{\pmb{x}=\pmb{G}\pmb{z}:\pmb{z}\in\mathbb{Z}^m\}$ be a lattice in $\mathbb{R}^m$ where $\pmb{G}$ is an $m\times m$ full-rank lattice generator matrix. Denote $Q_\Lambda(\pmb{x})=\arg\min_{\pmb{\lambda}\in\Lambda}|\pmb{\lambda}-\pmb{x}|$ as the nearest neighbour quantizer associated with a lattice $\Lambda$. The Voronoi cell that corresponds to the lattice point $\pmb{x}\in\Lambda_c$ is the set of points in $\mathbb{R}^m$ closest to $\pmb{x}$, i.e., $\mathcal{V}_{\pmb{x}}(\pmb{G})=\{\pmb{r}\in\mathbb{R}^m:Q_{\Lambda_c}(\pmb{r})=\pmb{x}\}$ and its volume is given by $V_c\stackrel{\Delta}{=}V(\mathcal{V}_{\pmb{x}}(\pmb{G}))=\sqrt{{\rm det}(\pmb{G}^{\mathsf{T}}\pmb{G})}$. The covering radius $r_{\rm cov}(\pmb{G})$ is the radius of the smallest sphere centered at the origin that contains $\mathcal{V}_{\pmb{0}}(\pmb{G})$. The packing radius $r_{\rm pack}(\pmb{G})$ of the lattice $\Lambda_c$ is the radius of the largest sphere centred at the origin that is within $\mathcal{V}_{\pmb{0}}(\pmb{G})$. The effective radius $r_{\rm eff}(\pmb{G})$ is the radius of the sphere with volume equal to $\mathcal{V}_{\pmb{0}}(\pmb{G})$. Let $\Lambda_s$ be a sublattice of $\Lambda_c$ with fundamental Voronoi region (corresponds to lattice point $\pmb{0}$) $\mathcal{V}_s$. An $m$-dimensional \textit{nested} lattice code is defined by $\mathcal{C}=\{\Lambda_c+\pmb{u}_0\}\cap\mathcal{V}_s$.

We say that a LAST code is nested if the underlying lattice code is nested. Here, the information message is effectively encoded into the cosets $\Lambda_s$ in $\Lambda_c$. As defined in \cite{Label2}, we shall call such codes the mod-$\Lambda$ scheme. The proposed mod-$\Lambda$ scheme works as follows. Consider the nested LAST code $\mathcal{C}$ defined by $\Lambda_c$ (the coding lattice) and by its sublattice $\Lambda_s$ (the shaping lattice) in $\mathbb{R}^{m}$. Assume that $\Lambda_s$ has a second-order moment $\sigma^2(\Lambda_s)=1/2$ (so that $\pmb{u}$ uniformly distributed over $\mathcal{V}_s$ satisfies $\mathsf{E}\{|\pmb{u}|^2\}=MTL$). The transmitter selects a codeword $\pmb{c}\in\mathcal{C}$, generates a dither signal\footnote{A dither signal is a random signal that is used to make the MMSE estimation error independent of the transmitted codeword (see \cite{EZ2} for further details).} $\pmb{u}$ with uniform distribution over $\mathcal{V}_s$, and computes $\pmb{x}=[\pmb{c}-\pmb{u}]\; {\rm mod}\; \Lambda_s$.

For the MIMO ARQ channel, we use the mod-$\Lambda$ \textit{incremental redundancy} scheme that was provided in \cite{Label1}.  The signal $\pmb{x}$ is partitioned into $L$ vectors of size $2MT$ each. Those vectors are transmitted, sequentially, in the different ARQ rounds based on the ACK/NACK feedback. Upon completion of the $\ell<L$ transmission, the receiver attempts to decode the message using lattice stack sequential decoder (see Section III) that implements a sort of deadline algorithm. In particular, the received signal, $\pmb{y}_{\ell}$, is multiplied by the forward filter matrix $\pmb{F}_{\ell}$ of the minimum mean-square error decision feedback equalization (MMSE-DFE) corresponding to the truncated matrix $\pmb{H}_{\ell}$, and then the dither signal filtered by the upper triangular feedback filter matrix $\pmb{B}_{\ell}$ of the MMSE-DFE is added to it (the definitions and some useful properties of the MMSE-DFE matrices $\pmb{F}$, $\pmb{B}$ are given in \cite{Label2}). In this case, the received signal can be expressed as
\begin{equation}\label{estimation1}
\pmb{y}'_{\ell}=\pmb{F}_{\ell}\pmb{y}_{\ell}+\pmb{B}_{\ell}\pmb{u}=\pmb{B}_{\ell}\pmb{c}'+\pmb{e}',
\end{equation}
where $\pmb{c}'=\pmb{c}+\pmb{\lambda}$, $\pmb{\lambda}=-Q_{\Lambda_s}(\pmb{c}-\pmb{u})$, and $\pmb{e}'=-[\pmb{B}_{\ell}-\pmb{F}_{\ell}\pmb{H}_{\ell}]\pmb{x}+\pmb{F}_{\ell}\pmb{w}_{\ell}$. One can easily verify (see \cite{Label2}) the relationship between $\pmb{B}_{\ell}$ and the channel matrix $\pmb{H}_{\ell}$ through the following equations:

For the case of long-term static channel, we have
\begin{equation}\label{BL}
\det(\pmb{B}_{\ell}^{\mathsf{T}}\pmb{B}_{\ell})=\left(\det\left(\pmb{I}+{\rho\over M}(\pmb{H}^c)^\mathsf{H}\pmb{H}^c\right)\right)^{2T\ell},
\end{equation}
and for the short-term static channel, we have
\begin{equation}\label{BL2}
\det(\pmb{B}_{\ell}^{\mathsf{T}}\pmb{B}_{\ell})=\prod\limits_{j=1}^{\ell}\det\left(\pmb{I}+{\rho\over M}(\pmb{H}_j^c)^\mathsf{H}\pmb{H}_j^c\right)^{2T}.
\end{equation}

The basic idea in this approach is to use a modified lattice stack sequential decoder for joint error detection and correction. The decoder first check if the channel is in outage. In this case, an error is declared and a NACK is sent back. If not, we use the modified lattice sequential decoder to find a lattice point that satisfies a certain predetermined condition (e.g., a time-out limit). Now if no point is found, an error is declared, and hence, a NACK bit is fed back. If a point is found to satisfy such condition then we proceed to the next step to find the codeword as $\hat{\pmb{c}}=[\pmb{G}\hat{\pmb{z}}]\;{\rm mod}\;\Lambda_s.$ The only exception to this rule is at the $L$-th ARQ round, where the regular lattice stack sequential decoder is used to find the closest lattice point.

\subsection{Rate and Reliability}
Let $\eta$ be defined as the average throughput of the ARQ scheme, expressed in transmitted bits per channel use. Following the definition used in \cite{Label1}, $\eta$ can be expressed as
\begin{equation}\label{eta}
\eta={R_1\over 1+\sum_{\ell=1}^{L-1}p(\ell)},
\end{equation}
where $R_1$ denotes the rate of the first block in bits per channel use, $p(\ell)=\Pr(\overline{\mathcal{A}_1},\dots,\overline{\mathcal{A}_\ell})$ with $\mathcal{A}_{\ell}$ denoting the event that an ACK is fed back at round $\ell$.

Let $E_{\ell}$ denotes the event that the transmitted message is incorrectly decoded by the ARQ decoder, then the probability of error can be upper bounded as
\begin{equation} \label{Pe}
P_e\leq\sum\limits_{\ell=1}^{L-1}\Pr(E_{\ell},\mathcal{A}_{\ell})+\Pr(E_L),
\end{equation}
where $\Pr(E_{\ell},\mathcal{A}_{\ell})$ takes the definition of the probability of \textit{undetected error} at round $\ell\leq L-1$.

Define respectively the \textit{effective} ARQ multiplexing gain and \textit{effective} ARQ diversity gain as
$$r_e=\lim_{\rho\rightarrow \infty}{\eta(\rho)\over\log\rho},\quad d=-\lim_{\rho\rightarrow\infty}{\log P_e(\rho)\over\log\rho}.$$

Achieving the optimal diversity-multiplexing-delay tradeoff in \cite{Label1} was performed using an incremental redundancy ARQ scheme coupled with a list lattice decoder for joint error detection and correction. This decoder (corresponds to~(\ref{estimation1})) finds all lattice points that satisfy (see \cite{Label1})
$$\left\{\pmb{x}\in\mathbb{R}^m: |\pmb{y}'-\pmb{B}_\ell\pmb{x}|^2\leq MTL(1+\gamma\log(\rho))\right\},$$
where $\gamma$ is a constant chosen appropriately to ensure the achievability of the optimal tradeoff. This decoder can be efficiently implemented using sphere decoding algorithms (see for example \cite{Label10}). Sphere decoders, however, are computationally very complex especially for low-to-moderate SNR and large signal dimensions where the output of the list sphere decoder can become extensively large that may result in a waste of time trying to decode the message. Hence, it is of great interest to search for a low complexity joint detector and decoder that can achieve the optimal tradeoff. This fact motivates us to replace the list lattice decoder by a more efficient retransmission strategy using lattice sequential decoders \cite{Label7}, \cite{WD}. The strategy is based on the stack algorithm and is designed to predict the occurrence of an error in advance by monitoring the number of computations performed by the decoder. Before we do that, we would like to introduce next the sequential decoder for lattice codes and some of its parameters that will be used to construct our new MIMO ARQ joint error detection and decoding scheme.

\section{Lattice Sequential Decoder: Performance Bounds and Complexity Distribution}
\subsection{Lattice Stack Algorithm}
The Stack sequential decoder is an efficient tree search algorithm that attempts to find a ``best fit'' with the received noisy signal. As in conventional stack decoder \cite{Label5}, to determine a best fit (path), values are assigned to each node on the tree. This value is called the metric. For the lattice stack sequential decoder, the metric at the $\ell$-th round [corresponds to (\ref{estimation1})] is given by (see~\cite{Label7})
\begin{equation}\label{metric}
\mu(\pmb{z}_1^k,\ell)=bk-|{\pmb{y}''_{\ell}}_1^k-\pmb{R}^{(\ell)}_{kk}\pmb{z}_1^k|^2,\quad \forall 1\leq k \leq m,
\end{equation}
where $\pmb{z}_1^k=[z_k,\cdots,z_2,z_1]^T$ denotes the last $k$ components of the integer vector $\pmb{z}\in\mathbb{Z}^m$, $\pmb{R}^{(\ell)}_{kk}$ is the lower $k\times k$ matrix of $\pmb{R}_{\ell}$ that corresponds to the QR decomposition of the code-channel matrix $\pmb{B}_{\ell}\pmb{G}=\pmb{Q}_{\ell}\pmb{R}_{\ell}$ at the $\ell$-th round, ${\pmb{y}''_{\ell}}_1^k$ is the last $k$ components of the vector ${\pmb{y}''_{\ell}}=\pmb{Q}_{\ell}^{\mathsf{T}}\pmb{y}'_{\ell}$, and $b\geq 0$ is the bias term. The bias parameter is critical for controlling the amount of computations required at the decoding stage.

As the decoder searches nodes, an ordered list of previously examined paths of different lengths is kept in storage. Each stack entry contains a path along with its metric. Each decoding step consists of extending the top (best) path in the stack. The determination of the best and next best nodes is simplified in the closest lattice point search problem by using the Schnnor-Euchner enumeration \cite{Label10} which generates nodes with metrics in ascending order given any node $\pmb{z}_1^k$. The decoding algorithm terminates when the top path in the stack reaches the end of the tree (refer to \cite{Label5} and \cite{Label7} for more details about the algorithm).

The main role of the bias term $b$ used in the algorithm is to control the amount of computations performed by the decoder. In this work, we define the computational complexity of the joint lattice sequential decoder as the total number of nodes visited by the decoder during the search, accumulated over all ARQ rounds, until a new transmission is started. Also, the bias term is responsible for the excellent performance-throughput-complexity tradeoff achieved by such decoding scheme. The role that the bias parameter plays in the new efficient decoding algorithm will be discussed in details in the subsequent sections.

\subsection{Performance Analysis: Lower and Upper Bounds}
In this section, we would like to derive lower and upper bounds on the error performance of the lattice sequential decoder. These bounds work as the primary elements for constructing our new efficient decoder for the MIMO ARQ channel.

Consider the detection at the $\ell$-th ARQ round. For simplicity, we consider here the long-term static channel (similar arguments can be done for the short-term static channel). Assume the received signal is $\pmb{y}_{\ell}=\pmb{B}_{\ell}\pmb{x}+\pmb{e}_\ell$, and denote $E_{ld}(\pmb{B}_\ell)$ and $E_{sd}(\pmb{B}_\ell,b)$ as the events that lattice decoder and lattice sequential decoder make an erroneous detection, respectively, where $b$ is the bias term that was introduced in (\ref{metric}). It is well-known \cite{WD} that lattice decoders outperform lattice sequential decoders for any $b>0$. In fact, the performance of the lattice decoder serves as a lower bound of the lattice stack sequential decoder. Now, lattice decoding disregard the boundaries of the lattice code and find the point of the underlying (infinite) lattice closest to the received point. As such, due to lattice symmetry, one can assume that the all-zero lattice point is transmitted. For a given lattice $\Lambda_c$, we have
 \begin{equation}\label{UB201}
 \begin{split}
P(E_{ld}(\pmb{B}_\ell)|\Lambda_c)={\rm Pr}\left(\bigcup_{\pmb{x}\in\Lambda_c^*}\left\{{2(\pmb{B}_\ell\pmb{x})^\mathsf{T}\pmb{e}_\ell}\geq{|\pmb{B}_\ell\pmb{x}|^2}\right\}\right)\leq P(E_{sd}(\pmb{B}_\ell,b)|\Lambda_c),
\end{split}
\end{equation}
where $\Lambda_c^*=\Lambda_c\backslash\{\pmb{0}\}$.

For the lattice sequential decoder, it seems a bit difficult to obtain an exact expression for its decoding error probability. Instead, we seek to derive an upper bound for the error performance of such a decoder which can be done as follows:
 \begin{equation}\label{UB202}
 \begin{split}
P(E_{sd}(\pmb{B}_\ell,b)|\Lambda_c)&\stackrel{(a)}{\leq}{\rm Pr}\left(\bigcup_{\pmb{z}\in\mathbb{Z}^m\backslash\{\pmb{0}\}}\{\mu(\pmb{z},\ell)>\mu_{\min}(\ell)\}\right)\\
&\stackrel{(b)}{\leq}{\rm Pr}\left(\bigcup_{\pmb{x}\in\Lambda_c^*}\{|\pmb{B}_\ell\pmb{x}|^2-2(\pmb{B}_\ell\pmb{x})^{\mathsf{T}}\pmb{e}_\ell<bm\}\right)\\
&={\rm Pr}\left(\bigcup_{\pmb{x}\in\Lambda_c^*}\left\{{2(\pmb{B}_\ell\pmb{x})^\mathsf{T}\pmb{e}_\ell}>{|\pmb{B}_\ell\pmb{x}|^2}\left(1-{bm\over |\pmb{B}_\ell\pmb{x}|^2}\right)\right\}\right),
\end{split}
\end{equation}
where $(a)$ is due to the fact that in general, $\mu(\pmb{z},\ell)>\mu_{\min}(\ell)$ is just a necessary condition for $\pmb{x}=\pmb{Gz}$ to be decoded by the stack decoder, $\mu_{\min}=\min\{0,b-|{\pmb{e}'_\ell}_1^1|^2,2b-|{\pmb{e}'_\ell}_1^2|^2,\ldots,bm-|{\pmb{e}'_\ell}_1^m|^2\}$ is the minimum metric that corresponds to the transmitted path with $\pmb{e}'_\ell=\pmb{Q}^\mathsf{T}\pmb{e}_\ell$, and $(b)$ follows by noticing that $-(\mu_{\min}+|{\pmb{e}'_\ell}|^2)\leq 0$. Following the footsteps of \cite{WD}, one can show that (\ref{UB202}) can be finally upper bounded as
\begin{equation}\label{UB300}
P(E_{sd}(\pmb{B}_\ell,b)|\Lambda_c)\leq P(E_{ld}(\tilde{\pmb{B}}_\ell)|\Lambda_c),
\end{equation}
where
\begin{equation}\label{newB2}
\tilde{\pmb{B}}_{\ell}=\left(1-{b\over 2^{[R_{\rm mod}(\ell)-R]/M}\phi(\ell)}\right)\pmb{B}_{\ell},
\end{equation}
where $R_{\rm mod}(\ell)$ is the rate at round $\ell$ that can be achieved using MMSE-DFE lattice decoding and according to (\ref{BL}) is given by
$$R_{\rm mod}(\ell)=\log\det(\pmb{B}_{\ell}^\mathsf{T}\pmb{B}_{\ell})^{1/2T}=\log\det\left(\pmb{I}+{\rho\over M}(\pmb{H}^c)^\mathsf{H}\pmb{H}^c\right)^{\ell},$$
$R$ is the transmission rate, and $\phi(\ell)=0.5(2r_{\rm pack}(\pmb{B}_{\ell}\pmb{G})/r_{\rm eff}(\pmb{B}_{\ell}\pmb{G}))^2$ . Interestingly, one may show that $\phi(\ell)$ is lower bounded by a constant independent of SNR and $\ell$ and as a result it has no effect on the performance in the SNR scale of interest.

It is clear from the above analysis that the lattice stack sequential decoder approaches the performance of the lattice decoder as $b\rightarrow 0$, i.e., $E_{sd}(\pmb{B}_\ell,0)=E_{ld}(\pmb{B}_\ell)$. Moreover, the upper bound (\ref{UB300}) corresponds to the probability of decoding error of a received signal $\pmb{y}_\ell=\tilde{\pmb{B}}_\ell\pmb{x}+\pmb{e}_\ell$ decoded using lattice decoding and is valid for all values of $b<2^{[R_{\rm mod}(\ell)-R]/M}\phi(\ell)$, i.e., $E_{sd}(\pmb{B}_\ell,b)=E_{ld}(\tilde{\pmb{B}}_\ell)$. Therefore, for a given lattice $\Lambda_c$, channel matrix $\pmb{H}_\ell$, and a bias term $b>0$, one can bound the error performance of the lattice sequential decoder as
\begin{equation}\label{bounds22}
P(E_{ld}(\pmb{B}_\ell)|\Lambda_c)\leq P(E_{sd}(\pmb{B}_\ell,b)|\Lambda_c)\leq P(E_{ld}(\tilde{\pmb{B}}_\ell)|\Lambda_c).
\end{equation}

By averaging (\ref{bounds22}) over the ensemble of random lattices $\Lambda_c$, one can show that (see \cite{WD}) for a fixed non-random channel matrix $\pmb{H}^c_{\ell}$, the rate
\begin{equation}\label{fano_rate}
R_b(\pmb{H}^c_{\ell},\rho)\triangleq \max\biggl\{R_{\rm mod}(\pmb{H}^c_{\ell},\rho)-2ML\log\left({1+\sqrt{1+8\alpha(\ell)}\over 2}\right),0\biggr\},
 \end{equation}
is achievable by LAST coding and MMSE-DFE lattice Fano/Stack sequential decoding with bias term $b\geq0$, where $\alpha$ is given by
\begin{equation}\label{alpha}
\alpha(\ell)=\left(r_{\rm eff}(\pmb{B}_{\ell}\pmb{G})\over 2r_{\rm pack}(\pmb{B}_{\ell}\pmb{G})\right)^2b.
\end{equation}
The equations (\ref{fano_rate}) and (\ref{alpha}) suggest that as long as the channel is well-conditioned, one may use large of values of bias term which is needed to achieve low decoding complexity (as will be discussed later). On the other hand, if the channel is close to outage, very low values of $b$ must be chosen in order to maintain high achievable rates. For example, if $b=0$, (\ref{fano_rate}) reduces simply to $R_b=R_{\rm mod}$, which corresponds to the rate achievable by the MMSE-DFE lattice decoder.

In general, for the long-term static channel, following the footsteps of \cite{WD}, one can show that if $b$ is allowed to vary with SNR and the channel statistics as
\begin{equation}\label{bias_SNR}
b(\pmb{\lambda},\rho)={1\over 2}{\prod_{i=1}^{M}(1+\rho\lambda_i)^{\ell/ML}\over \eta(\pmb{\lambda},\rho)^{\ell/ML}}\left[1-\left({\eta(\pmb{\lambda},\rho)\over \prod\limits_{i=1}^{M}(1+\rho\lambda_i)}\right)^{\ell/2ML}\right]\left(2r_{\rm pack}(\pmb{B}_{\ell}\pmb{G})\over r_{\rm eff}(\pmb{B}_{\ell}\pmb{G})\right)^2.
\end{equation}
In this case, one can easily show that by substituting $b$ in (\ref{fano_rate}) and (\ref{alpha}), we get
\begin{equation}
R_b(\pmb{\lambda},\rho)=\ell\log\eta(\pmb{\lambda},\rho).
\end{equation}
The term $\eta (\pmb{\lambda},\rho) $ can be chosen freely between 1 and  $\prod_{i=1}^{M}(1+\rho\lambda_i)$ (the maximum achievable rate under lattice decoding). Depending on the value of $\eta(\pmb{\lambda},\rho)$ we obtain different achievable rates and hence different outage performances.

We define the outage event under lattice sequential decoding as $\mathcal{O}_b(\rho,\ell)\triangleq\{\pmb{H}^c_{\ell}: R_b(\pmb{H}^c_{\ell},\rho)<R\}$. Denote $R=r\log\rho$. The probability that the channel is in outage, $P_{\rm out}(\ell,b)={\rm Pr}(\mathcal{O}_b(\rho,\ell))$, can be evaluated as follows:
  \begin{eqnarray}
  P_{\rm out}(\ell,b)={\rm Pr}(\ell\log\eta(\pmb{\lambda},\rho)<R)\limi{=}\rho^{-d_b(r/\ell)}.
  \end{eqnarray}
where $d_{b}(r)$ is the outage SNR exponent that is achieved by the MIMO channel with no ARQ (i.e., with $L=1$). For simplicity, we may express
\begin{equation}\label{eta_zeta}
\eta (\pmb{\lambda},\rho) =\phi\prod_{i=1}^{M}(1+\rho\lambda_i)^{\zeta_i},
\end{equation}
where $0<\phi<1$ is a constant independent of $\rho$, and $\zeta_i,\;\forall 1\leq i\leq M$, are constants that satisfy the following two constraints: $\sum_{i=1}^M \zeta_i\leq M$, and $\zeta_1\geq\zeta_2\geq\cdots\geq \zeta_M\geq 0$. It is not a simple matter to find a closed expression of $d_b(r)$ for all possible values of $\zeta_i$. However, one case of special interest, where a closed form of $d_b(r)$ can be easily derived, corresponds to $\zeta_i\in(0,M)$. In this case, one can show [see \cite{WD}, Theorem 4] that for the MIMO channel with no ARQ, $d_b(r)$ is the piecewise-linear function connecting the points $(r(k),d(k))$, $k=0,1,\cdots,M$, where
\begin{equation}\label{DMT2}
\begin{split}
r(0)&=0, \quad r(k)= \sum\limits_{i=M-k+1}^{M}\zeta_i,\;1\leq k\leq M,\\
d(k)&=(M-k)(N-k),  \quad\quad 0\leq k\leq M.
\end{split}
\end{equation}
The above choice of $\zeta_i$ leads to a bias term (assuming high SNR)
\begin{equation}\label{bias_SNR}
b(\pmb{\lambda},\rho)\approx{1\over 2}\prod_{i=1}^{M}(1+\rho\lambda_i)^{(1-\zeta_i)\ell/ML}
\end{equation}
which may grow exponentially with SNR as $\rho^\epsilon$, $\epsilon>0$, when the channel is not in outage. This may significantly reduce the decoding complexity. However, although the diversity at $r=0$ is not affected by the coefficients $\zeta_i\neq0$ ($d_b(0)=MN$), the more unbalanced the coefficients are, the worse the DMT is.  It is important to note that, other than the uniform assignment $\zeta_1=\zeta_2=\cdots=\zeta_M=1$, the optimal tradeoff
\begin{equation}\label{optimal_DMT}
d^*(r)=(M-r)(N-r), \quad \forall\; r\in[0,\min\{M,N\}),
\end{equation}
cannot be achieved. In this case, we have $b={1\over 2}\phi^{-1/M}[1-\phi^{1/2M}]$, i.e., $b$ is a constant\footnote{For the short-term static channel, although the channel changes from round to round, achieving the optimal tradeoff requires also the use of a constant bias term that is independent of the round $\ell$. Otherwise, for non-fixed $b$, the bias term must be adapted with the channel state at each ARQ round.} independent of $\rho$. In other words, one cannot let $b$ to scale with SNR as $\rho^\epsilon$ (to achieve very low decoding complexity) if the optimal DMT is to be achieved. Therefore, by varying the decoder parameter (bias term), one gets different performance-rate-complexity tradeoffs.

It is a simple matter to extend the above result to the $\ell$-th ARQ round. One can show that, there exists a sequence of full-dimensional LAST codes with block length $T\geq (M+N-1)/\ell$ that achieves the DMT curve $d_b(r_1/\ell)$ which is the piecewise-linear function connecting the points $(r(k),d(k/\ell))$, $k=0,1,\cdots,M$, where
\begin{equation}\label{DMT3}
\begin{split}
r(0)&=0, \quad r(k)= \sum\limits_{i=M-k+1}^{M}\zeta_i,\;1\leq k\leq M,\\
d(k/\ell)&=\left(M-{k\over\ell}\right)\left(N-{k\over\ell}\right),  \quad\quad 0\leq k\leq M.
\end{split}
\end{equation}
In this case, one can show that there exists a lattice code $\Lambda_c$ such that
\begin{equation}\label{E_sd1}
P(E_{sd}(\pmb{B}_\ell,b))\limi{=}\rho^{-d_b(r_1/\ell)}, \quad \forall\;1\leq\ell\leq L.
\end{equation}
under the condition $T\ell\geq N+M-1$.

It must be noted that, except for the case of $\ell=L$ (the final ARQ round), $d_b(r_1/\ell)$ is not the best achievable SNR exponent for the MIMO ARQ channel at any round $\ell<L$. This is because in the above analysis we have ignored the detection capabilities that the MIMO ARQ can have. By carefully designing the lattice sequential decoder to include an efficient error detection mechanism, we will show how the SNR exponent can be raised up to achieve the optimal tradeoff of the channel at any round $\ell$.

From a lattice point of view, the event of error under lattice decoding can be expressed as the event that the received signal is located outside the fundamental Voronoi region of the underlying (infinite) lattice. In this case, assuming $\pmb{0}$ was transmitted, one may express the bounds in (\ref{bounds22})
\begin{equation}\label{bounds2}
P(\pmb{e}\notin \mathcal{V}_{\pmb{0}}(\pmb{B}_\ell\pmb{G})|\Lambda_c)\leq P(E_{sd}(\pmb{B}_\ell,b)|\Lambda_c)\leq P(\pmb{e}\notin \mathcal{V}_{\pmb{0}}(\tilde{\pmb{B}}_\ell\pmb{G})|\Lambda_c).
\end{equation}
where $\tilde{\pmb{B}}_{\ell}$ is as defined in (\ref{newB2}). Interestingly, the upper bound in (\ref{bounds2}) provides us with the fact that as long as $\pmb{e}\in \mathcal{V}_{\pmb{0}}(\tilde{\pmb{B}}_\ell\pmb{G})$, the received signal is correctly decoded using lattice sequential decoding. As will be shown in the sequel, this fact can be used as the basic tool to construct our new efficient joint detection and decoding scheme.

\subsection{Computational Complexity Distribution}
An important parameter of the lattice sequential decoder is the distribution of computation, which characterizes the time needed to decode a message. It is well-known \cite{Label5} that the number of computations required to decode a message using sequential decoders is highly variable and assume very large values during intervals of high channel noise. Moreover, due to the random nature of the channel matrix and the additive noise, the computational complexity of such decoder is considered difficult to analyze in general. However, in \cite{WD} the authors have shown that, for the MIMO channel with no ARQ under lattice sequential decoding with constant bias term, the tail distribution becomes upper bounded by the asymptotic outage probability with SNR exponent that is equivalent to the optimal diversity-multiplexing tradeoff of the channel. This upper bound is shown to be achieved only of the number of computations performed by the decoder exceeds a certain limit when the channel is not in outage. This result can be easily extended to the MIMO ARQ channel as will be explained in the following.

Consider again decoding the received signal at the ARQ round $\ell=1,\cdots, L$. Let $\phi(\pmb{z}_1^k,\ell)$ be the indicator function defined by
\begin{equation}\label{phi}
\phi(\pmb{z}_1^k,\ell)=\begin{cases} 1, &\text{if node $\pmb{z}_1^k$ is extended;}\cr
                                                          0, &\text{otherwise,}\end{cases}
                                                          \end{equation}
and let $\mathcal{N}_j({\ell})$ be a random variable that denotes the total number of visited nodes during the search up to dimension $j$, at round $\ell<L$. In this case, $\mathcal{N}_j{(\ell)}$ can be expressed as
\begin{equation}\label{Nj}
\mathcal{N}_j{(\ell)}=\sum_{k=1}^{j}\sum_{\pmb{z}_1^k\in\mathbb{Z}^k}\phi(\pmb{z}_1^k,\ell).
\end{equation}

For the case of MIMO channel with no ARQ (i.e., $L=1$), following the footsteps of \cite{WD}, one can show that the \textit{asymptotic} computational complexity distribution of the MMSE-DFE lattice sequential decoder (assuming no error detection) is given by (for $T\geq N+M-1$)
\begin{equation}
\Pr(\mathcal{N}_j\geq \Gamma)\leq\Pr(\mathcal{N}_m\geq \Gamma)\limi{\leq} \rho^{-f(r)},
\end{equation}
for all $\Gamma$ that satisfy
\begin{equation}\label{L}
\Gamma\geq m+ \sum\limits_{k=1}^m{{(4\pi)}^{k/2} \over \Gamma(k/2+1)}{[bk+MT(1+\log\rho)]^{k/2}\over \det({\pmb{R}_{kk}^{\mathsf{T}}}\pmb{R}_{kk})^{1/2}},
\end{equation}
where $\pmb{R}_{kk}$ is the lower $k\times k$ part of the upper triangular matrix $\pmb{R}^{(\ell)}$ of the QR decomposition of $\pmb{B}_{\ell}\pmb{G}$, and $f(r)=(M-r)(N-r)$ for $r\in[0,\min\{M,N\})$.  It must be noted that (\ref{L}) is only valid for constant (fixed) bias term $b$, and a lower bound on $\Gamma$ for the general expression of $b$ (see (\ref{bias_SNR}) is known yet. However, as will be discuss in the sequel, even for fixed bias, the above lower bound shows a significant improvement in complexity compared to the more complex optimal sphere decoder.

The above important result indicates that there exists a finite probability that the number of computations performed by the decoder may become excessive even at high SNR, irrespective to the channel being ill or well-conditioned! This probability is usually referred to as the probability of a decoding failure. Such probability limits the performance of the lattice sequential decoder, especially for a one-way communication system. For a two-way communication system, such as in our MIMO ARQ system, the feedback channel can be used to eliminate the decoding failure probability. Therefore, our new decoder must be carefully designed to predict in advance the occurrence of decoding failure to avoid wasting the time trying to decode the message. This would result in a huge saving in decoding complexity. As will be shown in the sequel, the above result can be easily extended to the MIMO ARQ channel.

\section{Time-Out Algorithm}
It is well-known \cite{Label5} that the number of computations required to decode a message using sequential decoders is highly variable and assume very large values during intervals of high channel noise. As such, the decoder is expected to spend longer time attempting to decode the message. For the proposed incremental-redundancy MIMO ARQ system, this condition can be used as an indicator of when the receiver should terminate the search and request the transmitter for additional redundancy bits during any of the $\ell<L$ rounds.

In order to avoid wasting time trying to decode a noisy signal during any of $\ell<L$ ARQ rounds, we implement a \textit{time-out} algorithm in the lattice stack sequential decoder for joint error detection and correction. Such algorithm works as follows: we define a parameter $\Gamma_{\rm out}$ to be the maximum time (number of computations) allowed to decode a message during any of the $\ell<L$ ARQ rounds. If the decoding time exceeds $\Gamma_{\rm out}$, a NACK bit is fed back to the transmitter. The only exception of this rule is when the maximum number of ARQ rounds, $L$, is reached. In this case, the regular lattice sequential decoder (with no time-out limit) is used, where a NACK bit will be interpreted as an error, and the transmission of the next message is started anyway. Next, we define the \textit{retransmission probability} and the \textit{undetected error probability} from a lattice point of view. Those two quantities are responsible for the performance-throughput tradeoff achieved by the MIMO ARQ system. Throughout the work we assume the use of a small (fixed) bias term during all $L$ ARQ rounds. Before continuing our analysis, we would like to introduce some important definitions related to the MIMO ARQ channel that will be used throughout the paper.

Denote $0\leq\lambda^{(j)}_1\leq\cdots\leq\lambda^{(j)}_{\min\{M,N\}}$ the eigenvalues of $(\pmb{H}_j^c)^\mathsf{H}{\pmb{H}_j^c}$, $\forall 1\leq j\leq\ell$. Let us first define the outage event for both long-term and short-term static channels under the lattice sequential decoder with bias term $b$ that is given in (\ref{bias_SNR}) and $\eta(\pmb{\lambda},\rho)$ as defined in (\ref{eta_zeta}), with $\ell$ received blocks as
$$\mathcal{O}_b(\rho,\ell)=\biggl\{\pmb{H}_{j}^c\in\mathcal{C}^{N\times M}\;\forall 1\leq j\leq \ell\;:R_b(\rho)<R_1\biggr\},$$
where
\begin{equation}
R_b(\rho)=\begin{cases} \ell\sum\limits_{i=1}^M\zeta_i\log(1+\rho\lambda_i), & \textit{for long-term static channel;}\cr
\sum\limits_{j=1}^{\ell}\sum\limits_{i=1}^M\zeta_i(1+\rho\lambda_i^{(j)}),&\textit{for short-term static channel},
\end{cases}
\end{equation}
 Denote, $R_1=r_1\log\rho$, and define  $\alpha_i\triangleq-\log\lambda_i/\log\rho$, and $(x)^+=\max\{0,x\}$., then It can be easily verified that at high SNR, the outage event for the long-term static channel model can be expressed as (assuming $N\geq M$)
$$\mathcal{O}_{ls}(\ell)=\left\{\pmb{\alpha}\in\mathbb{R}^M_+: \alpha_1\geq\cdots\geq\alpha_M,\;\sum\limits_{i=1}^M\zeta_i(1-\alpha_i)^+<r_1/\ell\right\},$$
where in such channel we have $\lambda_i^{(j)}=\lambda_i$, $\forall 1\leq j\leq \ell$. For the short-term static channel we have
$$\mathcal{O}_{ss}(\ell)=\left\{(\pmb{\alpha}^{(1)},\cdots,\pmb{\alpha}^{(\ell)})\in\mathbb{R}^{M\ell}_+: \alpha^{(j)}_1\geq\cdots\geq\alpha^{(j)}_M,\forall1\leq j\leq \ell,\;\sum\limits_{j=1}^{\ell}\sum\limits_{i=1}^M\zeta_i\left[1-\alpha^{(j)}_i\right]^+<r_1\right\}.$$
Then, the associated asymptotic outage probability is given by (see \cite{Label1})
\begin{equation}
\outage{\ell}\limi{=}\begin{cases}\Pr(\mathcal{O}_{ls}(\ell))\limi{=} \rho^{-d_{\rm out}(\ell)}, & \textit{for long-term static channel}; \cr
\Pr(\mathcal{O}_{ss}(\ell))\limi{=} \rho^{-\ell d_{\rm out}(\ell)}, & \textit{for short-term static channel},
\end{cases}
\end{equation}
where $d_{\rm out}(\ell)=d_b(r_1/\ell)$ where $d_b(r)$ is as defined in (?).

\subsection{Retransmission Request Probability}
We make use of the lower and the upper bounds of the lattice sequential decoder's error performance in (\ref{bounds2}) to implement an error control mechanism for our IR-LAST MIMO ARQ system. It is clear from (\ref{newB2}) and (\ref{bounds2}) that $\mathcal{V}_{\pmb{0}}(\tilde{\pmb{B}}_\ell\pmb{G})\subseteq\mathcal{V}_{\pmb{0}}(\pmb{B}_\ell\pmb{G})$ for all $b\geq 0$. Therefore, at the decoder side, one can divide each of the Voronoi regions of the channel-code lattice $\Lambda(\pmb{B}_\ell\pmb{G})$, i.e., $\mathcal{V}_{\pmb{u}}(\pmb{B}_\ell\pmb{G})$ (corresponds to a lattice point $\pmb{u}=\pmb{B}_\ell\pmb{Gz}$, $\pmb{z}\in\mathbb{Z}^m$) into two disjoint regions --- $\mathcal{R}_{\pmb{u}}(\tilde{\pmb{B}}_\ell\pmb{G})$ and $\mathcal{V}_{\pmb{u}}(\pmb{B}_\ell\pmb{G})\backslash\mathcal{R}_{\pmb{u}}(\tilde{\pmb{B}}_\ell\pmb{G})$. This is depicted in Fig.~1. For convenience, we define $\mathcal{V}_{\pmb{u}}(\ell)=\mathcal{V}_{\pmb{u}}(\pmb{B}_{\ell}\pmb{G})$ and $\mathcal{R}_{\pmb{u}}(\ell)=\mathcal{R}_{\pmb{u}}(\tilde{\pmb{B}}_{\ell}\pmb{G})$.

Now, denote the region $\mathcal{D}(\ell)$ as
\begin{equation}
\mathcal{D}(\ell)=\mathbb{R}^m\backslash\left\{\bigcup_{\pmb{u}\in\Lambda(\pmb{B}_\ell\pmb{G})}\mathcal{R}_{\pmb{u}}(\ell)\right\}.
\end{equation}
We take advantage of the feedback channel by introducing the erasure option at the decoder such that whenever the received signal $\pmb{y}_{\ell}\in\mathcal{D}(\ell)$, for $\ell< L$, the decoder requests for a repeat transmission or additional bits (as in the case of IR-LAST coding scheme). In this case, at round $\ell< L$, the probability of a retransmission when the channel is \textit{not} in outage is given by
\begin{equation}\label{retrans}
\Pr(\overline{\mathcal{A}}_{\ell},\hbox{no outage})=\Pr(\pmb{y}_{\ell}\in\mathcal{D}(\ell)).
\end{equation}
Unfortunately, determining whether $\pmb{y}_{\ell}\in\mathcal{D}(\ell)$ or not is considered by itself a difficult problem. We try to simplify this problem through the use of lattice sequential decoding by tracking the number of computations performed during the search for the closest lattice point.

Following the definition of the number of computations performed by the decoder provided in (\ref{Nj}), a \textit{retransmission} is requested by the time-out algorithm at round $\ell<L$ if the number of computations exceeds the maximum time allowed before reaching the end of the tree. In other words, a NACK is sent back to the transmitter if, at any $1\leq j< m$, $\mathcal{N}_j{(\ell)}>\Gamma_{\rm out}$. This event could occur during a high channel noise period so that the received signal $\pmb{y}_\ell$ is close to the boundaries of a Voronoi cell of the lattice $\Lambda(\pmb{B}_\ell\pmb{G})$ and as a result, the decoder declares that $\pmb{y}_\ell\in\mathcal{D}(\ell)$. In this case, for the selected value of $b$, one has to carefully choose the time-out parameter $\Gamma_{\rm out}$ so that whenever $\mathcal{N}_j{(\ell)}>\Gamma_{\rm out}$, the decoder decides that $\pmb{y}_{\ell}\in\mathcal{D}(\ell)$ (see Fig.~2.(a)). Selecting an inappropriate value of $\Gamma_{\rm out}$ may result in the loss of the optimal tradeoff. Later, we shall make use of the following result for evaluating (\ref{retrans}), for fixed\footnote{The general case of a variable bias term that is given in (\ref{bias_SNR}) will not be considered here due to the difficulty of obtaining a lower bound on $\Gamma_{\rm out}$ for such bias values.} bias values:
\begin{lem}
For the long-term static ARQ channel, the asymptotic tail distribution of the total computational complexity of the lattice sequential decoder with fixed bias $b>0$, at round $\ell$ given the channel is not in outage, $\Pr(\mathcal{N}_j(\ell)\geq \Gamma_{\rm out})$, can be upper bounded by
\begin{equation}
\Pr(\mathcal{N}_j(\ell)\geq \Gamma_{\rm out})\limi{\leq}\rho^{-d_{\ell}},\quad \forall 1\leq j\leq m,
\end{equation}
under the condition
 $$\Gamma_{\rm out}\geq m+ \sum\limits_{k=1}^m{{(4\pi)}^{k/2} \over \Gamma(k/2+1)}{[bk+MTL(1+\zeta\log\rho)]^{k/2}\over \det(\pmb{R}_{kk}^{(\ell)\mathsf{T}}\pmb{R}^{(\ell)}_{kk})^{1/2}},$$
where $\zeta$ is a constant chosen sufficiently large enough so that $MTL\zeta\geq (M-r_1/\ell)(N-r_1/\ell)$,
\begin{equation}
\begin{split}
d_{\ell}&=\inf_{\pmb{\alpha}\in\overline{\mathcal{O}}_{ls}(\ell)}\biggl\{\sum_{i=1}^M (2i-1+N-M)\alpha_i\\
&\quad\qquad+T\ell\left[\sum_{i=1}^M(1-\alpha_i)^+-r_1/\ell\right]\biggr\},
\end{split}
\end{equation}
and
$$\overline{\mathcal{O}}_{ls}(\ell)=\left\{\pmb{\alpha}\in\mathbb{R}_+^M, \alpha_1\geq \cdots\geq \alpha_M:\sum_{j=1}^M (1-\alpha_i)^+\geq r_1/\ell\right\}.$$
\end{lem}
\begin{proof}
see Appendix I.
\end{proof}
It must be noted that the above result can be easily extended for the short-term static channel. In this case, one can show that $\Pr(\mathcal{N}_j(\ell)\geq \Gamma_{\rm out})\limi{\leq}\rho^{-\ell d_{\ell}},\; \forall 1\leq j\leq m$. Next, we derive an upper bound for the undetected error probability.

\subsection{Undetected Error Probability}
Another important parameter that characterizes the performance of the MIMO ARQ system, is the undetected error probability which was defined in (\ref{Pe}). In the IR-LAST MIMO ARQ system, an ACK bit is sent back to the transmitter if the decoder correctly decode the received signal subject to the condition that $\mathcal{N}_m(\ell)<\Gamma_{\rm out}$. This corresponds to the event $\pmb{y}_\ell\in\mathcal{R}_{\pmb{0}}(\ell)$, assuming $\pmb{0}$ was transmitted (see Fig.~2.(b)).  Also, an ACK is sent back to the transmitter, if decoding fails but it is not detected. Such an event occurs when the total number of computations $\mathcal{N}_m(\ell)<\Gamma_{\rm out}$, but the decoded lattice point is not $\pmb{0}$. This happens when the received signal $\pmb{y}_\ell\in\mathcal{R}_{\pmb{u}}(\ell)$ for any $\pmb{u}\neq\pmb{0}$ (see Fig.~2.(c)). Using the above argument, the undetected error probability can be expressed as (assuming $\pmb{0}$ was transmitted)
\begin{equation}\label{undetected}
\Pr(E_{sd}(\pmb{B}_{\ell},b),\mathcal{A}_{\ell})=\Pr\left(\bigcup_{\pmb{u}\in\Lambda^*(\pmb{B}_\ell\pmb{G})} \{\pmb{e}'_\ell\in\mathcal{R}_{\pmb{u}}(\ell)\}\right).
\end{equation}

Our goal now is to upper bound (\ref{undetected}). In this work, we will resort to a geometrical approach to obtain an upper bound on the undetected error probability. Before doing so, we express the event of sending an ACK, i.e., $\mathcal{A}_{\ell}$, in terms of the number of computations performed by the lattice sequential decoder as $\mathcal{A}_{\ell}=\{\mathcal{N}_m{(\ell)}< \Gamma_{\rm out}\}$. Consider now the following theorem:

\begin{thm}
For any lattice code $\Lambda_c$, the undetected error probability of the quasi-static $M\times N$ MIMO ARQ system with maximum rounds $L$, and codeword length $T$, under the time-out MMSE-DFE lattice sequential decoding scheme with parameters $b$ and $\Gamma_{\rm out}$, is bounded above as

In the case of long-term static channel
\begin{equation}
\Pr(E_{sd}(\pmb{B}_{\ell},b),\mathcal{A}_{\ell})\limi{\leq}\rho^{-f(r_1/L)},
\end{equation}
which is achieved with code block length $T\geq \lceil{(M+N-1)/L}\rceil$, where $f(r)=(M-r)(N-r)$, and any $b$ that satisfies (\ref{bias_SNR}).

In the case of short-term static channel
\begin{equation}
\Pr(E_{sd}(\pmb{B}_{\ell},b),\mathcal{A}_{\ell})\limi{\leq}\rho^{-Lf(r_1/L)},
\end{equation}
which is achieved with code block length $T\geq M+N-1$.
\end{thm}
\begin{proof}
It seems a bit difficult to obtain an upper bound directly for the undetected error probability using (\ref{undetected}). Therefore, we resort to a geometrical approach (see Fig.~3) to further upper bound (\ref{undetected}) by selecting $b$ such that the bound $r_{\rm eff}(\tilde{\pmb{B}}_{\ell}\pmb{G})\leq r_{\rm pack}(\pmb{B}_{\ell}\pmb{G})$ is maintained for all $\ell<L$. The effective radius $r_{\rm eff}(\tilde{\pmb{B}}_{\ell}\pmb{G})$ of the lattice generated by $\tilde{\pmb{B}}_{\ell}\pmb{G}$ is given by
\begin{equation}
r_{\rm eff}(\tilde{\pmb{B}}_{\ell}\pmb{G})=\left[{V_c\det({\tilde{\pmb{B}}_{\ell}}^\mathsf{T}\tilde{\pmb{B}}_{\ell})^{1/2}\over V(\mathcal{S}_m(1))}\right]^{1/m}.
                     \end{equation}
In this case, we can upper bound the undetected error probability as
\begin{equation}\label{UDE}
\Pr(E_{sd}(\pmb{B}_{\ell},b),\mathcal{A}_{\ell})=\Pr\left(\bigcup_{\pmb{u}\in\Lambda^*(\pmb{BG})} \{\pmb{e}'_{\ell}\in\mathcal{R}_{\pmb{u}}(\ell)\}\right)\leq\Pr(|\pmb{e}'_{\ell}|^2\geq r_{\rm eff}^2(\tilde{\pmb{B}}_{\ell}\pmb{G})).
\end{equation}\label{eps}
It is not yet clear how the RHS of (\ref{UDE}) can be evaluated. To overcome this problem, we can find a lower bound on $r_{\rm eff}^2(\tilde{\pmb{B}}_{\ell}\pmb{G})$ at high SNR as follows:

Asymptotically, one can express bias term that is defined by (\ref{bias_SNR}) as
\begin{equation}\label{p2}
b\limi{=}{\rho^{{\ell\over ML}\sum_{i=1}^M(1-\alpha_i)^+}\over \eta^{\ell/ML}}\left[1-\left({\eta\over \rho^{\sum_{i=1}^M(1-\alpha_i)^+}}\right)^{\ell/2ML}\right].
\end{equation}
Substituting (\ref{p2}) in (\ref{newB2}), when the channel is not in outage, one can upper bound $\det(\tilde{\pmb{B}}_{\ell}^\mathsf{T}\tilde{\pmb{B}}_{\ell}))\geq \eta^{2T}$, where $\eta\limi{=}\rho^{\sum_{i=1}^M\zeta_i(1-\alpha_i)}$. In this case, we have that
\begin{eqnarray} \label{reff}
r_{\rm eff}^2(\tilde{\pmb{B}}_{\ell}\pmb{G})&\geq&\left[{V_c\over V(\mathcal{R})}{V(\mathcal{R})\over V(\mathcal{S}_m(1))}\eta^{T}\right]^{2/m}\cr
&\stackrel{(a)}{\geq}&MTL\left[\rho^{-r_1T }\rho^{T\sum_{i=1}^M\zeta_i(1-\alpha_i)}\right]^{2/m}\cr
&\limi{\geq}&MTL \rho^{\nu}\limi{\geq}MTL(1+\gamma\log\rho),
\end{eqnarray}
where $(a)$ follows from the fact that (see \cite{Label9}) there exists a shifted lattice code $\Lambda_c+\pmb{u}_0^*$ with number of codewords inside the shaping region,
\begin{equation}
|\mathcal{C}(\Lambda_c,\pmb{u}_0^*,\mathcal{R})|=2^{R_1T}=\rho^{r_1T}\geq{V(\mathcal{R})\over V_c}.
\label{M}
\end{equation}
Also, $\nu={\ell\over ML}[\sum_{j=1}^{M}\zeta_i(1-\alpha_j)^{+}-{r_1/ \ell}]>0$ when the channel is not in outage, and the last inequality follows from the fact that $\lim_{\rho\rightarrow\infty}(1+\gamma\log\rho)/\rho^{\nu}=0$ for any $\nu,\gamma>0$.
 Therefore,
\begin{eqnarray}\label{UDEP}
\Pr(E_{sd}(\pmb{B}_{\ell},b),\mathcal{A}_{\ell})\limi{\leq}\Pr(|\pmb{e}'_\ell|^2\geq MLT(1+\gamma\log\rho))\limi{\leq}\rho^{-2MTL\gamma}.
\end{eqnarray}
By choosing a large enough value of $\gamma$ such that $MTL\gamma\geq f(r_1/L)$, we obtain
\begin{equation}
\Pr(E_{sd}(\pmb{B}_{\ell},b),\mathcal{A}_{\ell})\limi{\leq}\rho^{-f(r_1/L)},
\end{equation}
under the condition that $LT\geq M+N-1$.

The above analysis also applies to the short-term static channel, and one can show that under the condition $T\geq N+M-1$, we have
\begin{equation}
\Pr(E_{sd}(\pmb{B}_{\ell},b),\mathcal{A}_{\ell})\limi{\leq}\rho^{-Lf(r_1/L)}.
\end{equation}
\end{proof}

As will be shown in the sequel, (\ref{retrans}) and (\ref{undetected}) play an important role in determining the achievable diversity-multiplexing-delay tradeoff of the MIMO ARQ channel. Reducing the number of retransmissions comes at the expense of increasing the undetected error probability, which is undesirable. It is clear that both probabilities are closely related and hence increasing or decreasing one of them may lead to a loss in the optimal tradeoff of the channel.

\subsection{Achieving the Optimal Tradeoff: Bias Term vs. $\Gamma_{\rm out}$}
Our goal here is to prove the optimality of the time-out lattice sequential decoder in terms of the achievable diversity-multiplexing-delay tradeoff. It is well-known \cite{Label7} that the bias term $b$ controls the amount of the computations performed by the lattice sequential decoder during the search and responsible for the excellent performance-complexity tradeoff achieved by such a decoder. Choosing a very large value of $b$ although greatly reduces decoding complexity, it may lead to a loss in the optimal tradeoff of the channel. Therefore, it is expected that the time-out parameter $\Gamma_{\rm out}$ will be a function of the bias term $b$ chosen in the algorithm. One may have already noticed that in Lemma~1. It turns out that an optimal value of $\Gamma_{\rm out}$, denoted by $\Gamma_{\rm out}^*$, exists so that the optimal tradeoff is achieved with a fairly low decoding complexity (i.e., average number of computations) compared to the joint list lattice decoder. The achievability of the optimal tradeoff, for any \textbf{\textit{fixed}} bias term, under time-out lattice sequential decoding is summarized in the following theorem:
\begin{thm}
Consider a MIMO ARQ channel under short-power constraint given in (\ref{power}), with $M$ transmit, $N$ receive antennas, a maximum number of ARQ rounds $L$, an effective multiplexing gain $0\leq r_e<\min\{M,N\}$. Then, the IR-LAST coding scheme under \textit{time-out} lattice stack sequential decoding with parameter $\Gamma_{\rm out}$ and fixed $b>0$, achieves the optimal tradeoff:

In the case of long-term static channel
$$d_{ls}^{*}(r_e,L)=\begin{cases}f\left(\displaystyle{r_e\over L}\right), & 0\leq r_e<\min\{M,N\};\\
                                                   0, & r_e\geq\min\{M,N\},\end{cases}$$
which is achieved with code block length $T\geq \lceil{(M+N-1)/L}\rceil$, where $f(r)=(M-r)(N-r)$.

In the case of short-term static channel
$$d_{ss}^{*}(r_e,L)=\begin{cases}Lf\left(\displaystyle{r_e\over L}\right), & 0\leq r_e<\min\{M,N\};\\
                                                   0, & r_e\geq\min\{M,N\},\end{cases}$$
which is achieved with code block length $T\geq M+N-1$. The optimal tradeoff is achieved subject to the condition
\begin{equation}\label{gamma}
\Gamma_{\rm out}\geq m+ \sum\limits_{k=1}^m{{(4\pi)}^{k/2} \over \Gamma(k/2+1)}{[bk+MTL(1+\zeta\log\rho)]^{k/2}\over \det(\pmb{R}_{kk}^{(\ell)\mathsf{T}}\pmb{R}^{(\ell)}_{kk})^{1/2}}.
\end{equation}
\end{thm}
\begin{proof}
see Appendix III.
\end{proof}

It must be noted that the above theorem is \textit{only} valid for non-zero, but fixed values of $b$. Although one can fully characterize the achievable tradeoff in terms of the bias term\footnote{For example, if we let $b$ to vary with SNR and channel statistics as given in (\ref{bias_SNR}) one can achieve an asymptotic SNR exponent $d^*_{ls}(r_e,L)=d_b(r_e/L)$ and $d_{ss}^*(r_e,L)=Ld_b(r_e/L)$ for the long-term and short term static channels, respectively, where $d_b(r)$ is as described in section III.} (see Section III), for non-fixed bias, it is not yet clear how the time-out parameter $\Gamma_{\rm out}$ changes with the bias term for variable $b$. Therefore, in what follows, and for the purpose of completing the analysis, we will only consider the case of fixed bias term.

Now, since $\Gamma_{\rm out}$ depends on the channel statistics (i.e., it is random), it would be desirable to determine its average value (averaged over channel statistics when it is not in outage). This may shed the light on determining the optimal value of $\Gamma_{\rm out}$ that can be used to achieve the optimal tradeoff of the channel. This can be done as follows:

Consider the long-term static channel. Therefore, the optimal average value of $\Gamma_{\rm out}$, say $\overline{\Gamma^*}_{\rm out}$, may be asymptotically lower bounded by
\begin{equation}\label{Gamma_av}
\overline{\Gamma^*}_{\rm out}\limi{=}m+\sum\limits_{k=1}^m(\log\rho)^{k/2}\mathsf{E}_{\pmb{\alpha}\notin\mathcal{O}_{ls}(\ell)}
\left\{\det(\pmb{R}_{kk}^{(\ell)\mathsf{T}}\pmb{R}^{(\ell)}_{kk})^{-1/2}\right\},
\end{equation}

In this paper we focus our analysis on nested LAST codes, specifically LAST codes that are generated using construction A that is described below (see \cite{Label9}).

We consider the Loeliger ensemble of mod-$p$ lattices, where $p$ is a prime. First, we generate the set of all lattices given by
$$\Lambda_p=\kappa (\mathsf{C}+p\mathbb{Z}^{2MTL})$$
where $p\rightarrow \infty$, $\kappa\rightarrow 0$ is a scaling coefficient chosen such that the fundamental volume $V_f=\kappa^{2MTL}p^{2MTL-1}=1$, $\mathbb{Z}_p$ denotes the field of mod-$p$ integers, and $\mathsf{C}\subset\mathbb{Z}_p^{2MTL}$ is a linear code over $\mathbb{Z}_p$ with generator matrix in systematic form $[\pmb{I}\;\pmb{P}^\mathsf{T}]^\mathsf{T}$. We use a pair of self-similar lattices for nesting. We take the shaping lattice to be $\Lambda_s=\zeta\Lambda_p$, where $\zeta$ is chosen such that the covering radius is $1/2$ in order to satisfy the input power constraint. Finally, the coding lattice is obtained as $\Lambda_c=\rho^{-r/2ML}\Lambda_s$. Interestingly, one can construct a generator matrix of $\Lambda_p$ as (see \cite{Conway})
\begin{equation}
\pmb{G}_p=\kappa\begin{pmatrix}
\pmb{I} & \pmb{0}\\
\pmb{P} & p\pmb{I}
 \end{pmatrix},
\end{equation}
which has a lower triangular form. In this case, one can express the generator matrix of $\Lambda_c$ as $\pmb{G}=\rho^{-r/2ML}\pmb{G}'$, where $\pmb{G}'=\zeta\pmb{G}_p$. Thanks to the lower triangular format of $\pmb{G}$. If $\pmb{M}$ is an $m\times m$ arbitrary full-rank matrix, where $m=2MTL$, and $\pmb{G}$ is an $m\times m$ lower triangular matrix, then one can easily show that
\begin{equation}\label{matrix_kk}
\det[(\pmb{MG})_{kk}]= \det(\pmb{M}_{kk})\det(\pmb{G}_{kk}),
\end{equation}
where $(\pmb{MG})_{kk}$, $\pmb{M}_{kk}$, and $\pmb{G}_{kk}$, are the lower $k\times k$ part of $\pmb{MG}$, $\pmb{M}$, and $\pmb{G}$, respectively.

Using the above result, one can express the determinant that appears in (\ref{Gamma_av}) as
\begin{equation}
\det(\pmb{R}_{kk}^{(\ell)\mathsf{T}}\pmb{R}^{(\ell)}_{kk})=
\det(\pmb{B}_{kk}^{(\ell)\mathsf{T}}\pmb{B}^{(\ell)}_{kk})
\det(\pmb{G}_{kk}^\mathsf{T}\pmb{G}_{kk})=\rho^{-rk/2ML}
\det(\pmb{R}_{kk}^{(\ell)\mathsf{T}}\pmb{R}^{(\ell)}_{kk})\det({\pmb{G}'}_{kk}^\mathsf{T}{\pmb{G}'}_{kk}),
\end{equation}
 Let $\mu_1\leq \mu_2\leq\cdots\leq\mu_k$ be the ordered nonzero eigenvalues of $\pmb{B}_{kk}^{(\ell)\mathsf{T}}\pmb{B}^{(\ell)}_{kk}$, for $k=1,\cdots,m$. Then,
$$\det(\pmb{B}_{kk}^{(\ell)\mathsf{T}}\pmb{B}^{(\ell)}_{kk})=\prod\limits_{j=1}^k\mu_j$$
 Note that for the special case when $k=m$ we have $\mu_{2(j-1)TL+1}=\cdots=\mu_{2jTL}=1+\rho\lambda_j((\pmb{H}^c)^\mathsf{H}\pmb{H}^c)$, for all $j=1,\cdots,M$.

Denote $\alpha'_i=-\log\mu_i/\log\rho$. Using (\ref{matrix_kk}), one can asymptotically express $\overline{\Gamma^*}_{\rm out}$ as
\begin{equation}
\overline{\Gamma^*}_{\rm out}\limi{=}m+\sum\limits_{k=1}^m(\log\rho)^{k/2}\mathsf{E}_{\pmb{\alpha}\notin\mathcal{O}_{ls}(\ell)}\{\rho^{c_k}\},
\end{equation}
where
\begin{equation}
c_k={1\over 2}\sum\limits_{j=1}^k \left({r_1\over ML} - \alpha'_j\right)^+.
\end{equation}
Now, since $c_k$ is non-decreasing in $k$, we have
\begin{equation}
\overline{\Gamma^*}_{\rm out}\limi{=}m+(\log\rho)^{m/2}\mathsf{E}_{\pmb{\alpha}\notin\mathcal{O}_{sl}(L)}\{\rho^{c_m}\},
\end{equation}
where
$$c_m=TL\sum\limits_{i=1}^M \left(\displaystyle{r_1\over ML}-(1-\alpha_i)^+\right)^+.$$

At multiplexing gain $r_1$, we have the channel is in outage only when $\sum_{j=1}^M(1-\alpha_j)^+\leq r_1/L$. In this case, we have
\begin{align*}
\mathsf{E}_{\pmb{\alpha}\notin\mathcal{O}_{sl}(L)}\{\rho^{c_m}\}&=\int\limits_{\pmb{\alpha}\notin\mathcal{O}_{sl}(L)}
\rho^{c_m} f_{\pmb{\alpha}}(\pmb{\alpha})\;d\pmb{\alpha}\\
&\limi{=}(\log\rho)^{m/2}\int\limits_{\pmb{\alpha}\notin\mathcal{O}_{sl}(L)}
\exp\left(\log\rho\left[TL\sum\limits_{i=1}^M\left({r_1\over ML}-(1-\alpha_i)^+\right)^+-\sum\limits_{i=1}^M (2i-1+N-M)\alpha_i\right]\right)\;d\pmb{\alpha}\\
&\limi{=}(\log\rho)^{m/2}\rho^{l(r_1)},
\end{align*}
where $\mathcal{O}_{sl}(L)=\left\{\pmb{\alpha}\in\mathbb{R}_+^M: \sum_{i=1}^{M}(1-\alpha_i)^{+}<r_1/L\right\}$, and
\begin{equation}\label{l'_r}
l(r_1)=\max_{\pmb{\alpha}\notin\mathcal{O}_{sl}(L)} \left[TL\sum\limits_{i=1}^M\left({r_1\over ML}-(1-\alpha_i)^+\right)^+-\sum\limits_{i=1}^M (2i-1+N-M)\alpha_i\right].
\end{equation}
It is not so difficult to see that the optimal channel coefficients that maximize (\ref{l'_r}) are
$$\alpha_i^*=1, \quad \hbox{for }i=1,\cdots,M-k,$$
and
$$\alpha_i^*=0, \quad \hbox{for }i=M-k+1,\cdots,M.$$
Substituting $\pmb{\alpha}^*$ in (\ref{l'_r}), we get
\begin{equation}
l(r_1)={Tr_1\over M}\left(M-{r_1\over L}\right)-\left(M-{r_1\over L}\right)\left(N-{r_1\over L}\right),
\end{equation}
for $r_1=0,1,\cdots,M$. An since $r_e\limi{=}r_1$, the asymptotic average computational complexity, when the channel is in outage, can be expressed as\footnote{As a reminder, the logarithm that appears in the complexity analysis is to the base 2.}
\begin{equation}\label{optimal_out}
\overline{\Gamma^*}_{\rm out}\limi{=}2MTL+(\log\rho)^{MTL}\rho^{l(r_e)}.
\end{equation}

One interesting special case of computing the optimal average time-out parameter $\overline{\Gamma^*}_{\rm out}$ is when $r_e=0$, i.e., when using a code with fixed rate $R_1$. In this case we have
\begin{equation}\label{optimal_out1}
\overline{\Gamma^*}_{\rm out}\limi{=}2MTL+{(\log\rho)^{MTL}\over \rho^{MN}}.
\end{equation}
Thus, the above equation describes how the time-out limit is related to the system parameters $(M,N,T,L,\rho)$ at the high SNR regime. As an example, consider a $2\times 2$ MIMO ARQ channel with $L=2$ rounds. Then, according to Theorem 2, we have $T=3$ is sufficient to achieve the optimal tradeoff. The signal dimension in this case is $m=24$. Assume $\rho=100$ (20 dB). According to (\ref{optimal_out}), the optimal time-out limit is given by $\overline{\Gamma^*}_{\rm out}\geq 98$. As will be shown in the sequel, this theoretical result closely matches the value of $\overline{\Gamma^*}_{\rm out}$ that is obtained experimentally (see Section V). Typical values of $b$ that corresponds to $\overline{\Gamma^*}_{\rm out}\approx 98$ are between $0.6$ and 1. For very small values of $b$, the average number of computations increases and according to that $\overline{\Gamma^*}_{\rm out}\gg 98$.

To see the great advantage of using the time-out lattice sequential decoder with constant bias term over the list lattice decoder implemented via sphere decoding algorithms, we compare the average computational complexity of both decoders. It has been shown in \cite{WD2} that, for moderate-to-high SNR, the average computations performed by the sphere decoder when the channel is not in outage, say $\Gamma_{\rm sphere}$ for a system with $m=2MTL$ signal dimension is given by (assuming fixed rate $r_1=0$)
\begin{equation}
\Gamma_{\rm sphere}\limi{=}2MTL+ {(\log\rho)^{2MTL}\over \rho^{MN}}.
\end{equation}
The ratio of the average complexity of both decoder, say $\gamma$, is given by
$$\gamma={\Gamma_{\rm sphere}\over {\overline{\Gamma^*}}_{\rm out}}\limi{=}{2MTL+(\log\rho)^{2MTL}/\rho^{MN}\over 2MTL + (\log\rho)^{MTL}/\rho^{MN}}.$$
This is a huge saving in computational complexity, especially for large signal dimensions even at very high SNR. For example, at $\rho=10^8$ (80 dB), $\gamma\approx 7.4$, i.e., the list sphere decoder's complexity is about 7 times larger than the complexity of the new proposed time-out lattice sequential decoder. For $\rho < 80$ dB, one would expect the ratio $\gamma\gg 7$. For extremely high SNR values (e.g., $\rho\geq 90$ dB), it seems that $\gamma\rightarrow 1$ as $\rho\rightarrow\infty$.

Moreover, some interesting remarks about the effect of extremely decreasing or increasing the value of $b$ on the performance-throughput-complexity tradeoff are discussed next. For very large SNR values, it is

It is well-known that as $b\rightarrow \infty$, lattice sequential decoders based on Schnnor-Euchner enumeration converts to the MMSE-DFE decoder \cite{Label7}. In terms of the total number of visited nodes, MMSE-DFE decoder achieves \textit{linear} computational complexity in $m$. In this case, for any $\Gamma_{\rm out}>m$, the message will always be decoded from the first round. Although it achieves high throughput and is computationally efficient, this decoder cannot achieve the optimal tradeoff. Assuming $N\geq M$, the maximum SNR exponent that such a decoder can achieve is $(N-M+1)(1-r_e/ML)^+$ (see \cite{MMSE2} for more details about MMSE decoding for the case of MIMO channel with no ARQ).

On the other hand, as $b\rightarrow 0$ the decoder achieves the best performance. However, the decoding complexity becomes equivalent to (and for some cases worst than) the complexity of lattice (sphere) decoding which is extensively large. Our main objective of using lattice sequential decoding is to save on decoding complexity. Therefore, one should appropriately select the bias term $b$ so that the optimal tradeoff is achieved while reducing the computational complexity. This may be achieved by the sequential decoder by ensuring that the path metric along the correct path increases on average, while decreases along other paths. In this case, we choose $b$ such that $\mathbb{E}_{\pmb{e}'}\{\mu(\pmb{z}_1^k)\}>0$ (assuming $\pmb{z}_1^m$ is the correct path). This corresponds to $b>\mathbb{E}\{|[\pmb{e}']_i|^2\}=1/2$. This fact is verified experimentally as will be shown in the next section.

As discussed earlier, the value of the parameter $\Gamma_{\rm out}$ used in the time-out algorithm is critical for achieving the optimal tradeoff of the MIMO ARQ system. In order to achieve the optimal tradeoff, both $b$ and $\Gamma_{\rm out}$ have to be appropriately selected in the time-out algorithm. Remember that the probability of error is upper bounded by
\begin{equation} \label{Pe1}
P_e\leq\sum\limits_{\ell=1}^{L-1}\underbrace{\Pr(E_{\ell},\mathcal{A}_{\ell})}_{\hbox{controlled by }\Gamma_{\rm out}}+\underbrace{\Pr(E_L)}_{\hbox{controlled by }b}
\end{equation}
where $\Pr(E_{\ell},\mathcal{A}_{\ell})\leq \rho^{-f(r_e/L)}$ is the probability of undetected error at round $\ell$ and is mainly controlled by the time-out parameter $\Gamma_{\rm out}$. However, the second term of the RHS of the above upper bound is affected by the value of the bias term $b$. Therefore, both $b$ and $\Gamma_{\rm out}$ have to be appropriately selected in the time-out algorithm so that a balance is obtained which may lead to achieving the optimal tradeoff of the channel.

Suppose that an optimal value of $b$ is set in the time-out algorithm. Then, choosing a very small or large value of $\Gamma_{\rm out}$ may result in a loss of the optimal tradeoff. This is because $\Gamma_{\rm out}$ can be seen as the parameter responsible for the amount of retransmission when error is detected. Choosing a large value of $\Gamma_{\rm out}$ reduces the probability of retransmission and may result in large performance degradation. In this case, as $\Gamma_{\rm out}\rightarrow \infty$, we have $\Pr(\mathcal{A}_{\ell})\rightarrow 1$. Therefore, the undetected error probability becomes equivalent to (for the long-term static channel)
$$
\Pr(E_{\ell},\mathcal{A}_{\ell})=\Pr(E_{\ell})\limi{=}\rho^{-f(r_1/\ell)}.
$$
Moreover, the probability of retransmission given the channel is not in outage approaches 0, which means that $r_e=r_1$ ($p(\ell)\rightarrow0$). And since the average error probability defined in (\ref{Pe}) is dominated by the term with the smallest SNR exponent, we have
\begin{equation}\label{bad}
P_e(\rho)\limi{\leq}\rho^{-f(r_1/L)}+\sum\limits_{\ell=1}^L\rho^{-f(r_1/\ell)}\limi{=}\rho^{-f(r_e)}.
\end{equation}
This is equivalent to the performance of the MIMO channel with no ARQ. On the other hand, choosing a very small value of $\Gamma_{\rm out}$ improves the performance\footnote{Performance improvement is achieved as a coding gain and not in the SNR exponent. In this case, the decoder will accumulate more information about the message before making a decision which is due to the fact that most of the time the decoder asks for retransmission.} at the expense of large throughput loss. In this case, the undetected error probability approaches 0, $p(\ell)\rightarrow 1$, and the throughput $\eta\rightarrow R_1/L$, i.e., $r_e=r_1/L$. The error probability at high SNR in this case is also given by~(\ref{bad}).

Therefore, in order to achieve the optimal tradeoff of the MIMO ARQ channel, first $b$ must be chosen to ensure that the achievability of the optimal diversity-multiplexing tradeoff when operating over the whole received signal (i.e., at round $L$). Then, $\Gamma_{\rm out}$ is selected accordingly so that the optimal diversity-multiplexing-delay tradeoff is achieved.

\section{Simulation Results}
In our simulation, we consider a long-term static MIMO ARQ link with $M=N=L=2$, $T=3$ and $R_1=8$ bits per channel use. The incremental redundancy LAST code is obtained as an $(m,p,k)$ Loeliger construction (refer to \cite{Label9} for a detailed description). The frame error rate and computational complexity are plotted in Fig.~4 and Fig.~5, respectively, for different values of $b$ used in the time-out algorithm. We measure the computational complexity of the joint lattice sequential decoder as the total number of nodes visited by the decoder during the search, accumulated over all ARQ rounds, until a new transmission is started. For every value of $b$, the optimal value of $\Gamma_{\rm out}$, denoted by $\Gamma_{\rm out}^*$, was found via simulation by trial and error. Using those optimal parameters, it is shown that the IR-LAST coding scheme decoded using the time-out lattice sequential decoder can achieve probability of error very close to the one that corresponds to the same IR-LAST coding scheme decoded using the list lattice decoder. This is achieved with significant reduction in complexity compared to the list lattice decoder (see Fig.~5). It is interesting to see how such IR-LAST coding scheme can achieve probability of error colose to the coherent LAST code with half the rate (4 bpcu). On the other hand, the effective rate, $R_e$, of the IR-LAST coding scheme decoded under the new proposed decoder is shown to approach $R_1=8$ as SNR grows as predicted by the theory. Optimal values of $\Gamma_{\rm out}$ for some special cases of $b$ are provided in Table 1. As expected, for values of $b<1/2$, the average computational complexity of the time-out algorithm increases and as a consequence, the value of $\Gamma_{\rm out}^*$ is proportionally increased. Simulation results demonstrate the excellent performance-complexity tradeoff achieved by the proposed algorithm for all values of $b$, especially at the moderate-to-high SNR regime (see Fig.~5).

{
\begin{table}[ht]
\caption{Optimum values of $\Gamma_{\rm out}$ for some special cases of $b$ used in the time-out algorithm for the case of $M=N=L=2$ and $T=3$ MIMO ARQ system using IR-LAST random code}
\begin{center}
\normalsize
\begin{tabular}{l*{2}{c}r}
\hline
$b$  & $\Gamma_{\rm out}^*$  \\
\hline
$0.6$ & 100   \\
$0.4$ & 800  \\
$0.1$  & $4\times 10^4$   \\
\hline
\end{tabular}
\end{center}
\end{table}
}

The error rates are obtained by averaging over at least 10~000 channel realizations at small SNRs and as much channel realization as required to count at least 100 frame errors at high SNRs.  It is clear that for values of $b>1/2$ the decoder exhibits some performance degradation compared to the incomplete list sphere decoder. To improve the performance one need to resort to smaller values of $b$ at the price of increasing computational complexity (see Fig.~4 and Fig.~5). In general and for all finite values of $b$, the time-out lattice stack sequential decoder has much lower computational complexity compared to the incomplete list sphere decoder.

\section{Conclusion}
In this paper, we have demonstrated, analytically and via simulation, the significant improvements achieved by the lattice stack sequential decoder over the incomplete list lattice decoder that is used for joint error detection and correction in the IR-LAST MIMO ARQ channel. A time-out algorithm has been proposed. Theoretical analysis and simulation results show that the optimal tradeoff can be achieved by such algorithm with very low decoding complexity compared to the list lattice decoder, especially for moderate-to-high SNR for which the list output could be extensively large.

\section*{Appendix I:\\ Proof of Lemma 1}
Let $\phi'(\pmb{z}_1^k,\ell)$ be the indicator function defined by
$$\phi'(\pmb{z}_1^k,\ell)=\begin{cases} 1, &\text{if $|{\pmb{e}''_{\ell}}_1^k-\pmb{R}^{(\ell)}_{kk}\pmb{z}_1^k|^2\leq bk-\mu_{\min}(\ell)$;}\cr
                                                          0, &\text{otherwise,}\end{cases}$$
where $\mu_{\min}(\ell)$ is the minimum metric along the decoded path. Then, it can be easily verified that
\begin{equation}\label{C221}
\sum_{\pmb{z}_1^k\in\mathbb{Z}^k}\phi(\pmb{z}_1^k,\ell)\leq\sum_{\pmb{z}_1^k\in\mathbb{Z}^k}\phi'(\pmb{z}_1^k,\ell),
\end{equation}
where $\phi(\pmb{z}_1^k,\ell)$ is as defined in (\ref{phi}), where a path may be extended by the stack decoder if the partial path metric at that node satisfies $\mu(\pmb{z}_1^k,\ell)\geq \mu_{\min}(\ell)$. Now, given $|\pmb{e}''_{\ell}|^2\leq MTL(1+\gamma)$, and by noticing that $-(\mu_{\min}(\ell)+|\pmb{e}''_{\ell}|^2)\leq 0$, we obtain
\begin{equation}
\sum_{\pmb{z}_1^k\in\mathbb{Z}^k}\phi'(\pmb{z}_1^k,\ell)\leq\sum_{\pmb{z}_1^k\in\mathbb{Z}^k}\phi^{''}(\pmb{z}_1^k,\ell),
\label{aa}
\end{equation}
where
\begin{equation}\label{C111}
\phi^{''}(\pmb{z}_1^k,\ell)=\begin{cases} 1, &\text{if $|{\pmb{e}''_{\ell}}_1^k-\pmb{R}^{(\ell)}_{kk}\pmb{z}_1^k|^2\leq bk+MTL(1+\gamma)$;}\cr
                                                          0, &\text{otherwise.}\end{cases}
                                                          \end{equation}
Notice the independence of the upper bound (\ref{aa}) on $\mu_{\min}(\ell)$. Now, let
$$
\phi^{'''}_k(\pmb{z},\ell)=\begin{cases} S_k(\ell), &\text{if $|{\pmb{e}''_{\ell}}-\pmb{R}_{\ell}\pmb{z}|^2\leq bm-\mu_{\min}(\ell)$;}\cr
                                                          0, &\text{otherwise,}\end{cases}
$$
where
\begin{equation}\label{Sk}
S_k(\ell)=\sum_{\pmb{z}_1^k\in\mathbb{Z}^k}\phi^{''}(\pmb{z}_1^k,\ell),
\end{equation}
then it can be easily shown that
$$\sum_{\pmb{z}_1^k\in\mathbb{Z}^k}\phi^{'}(\pmb{z}_1^k,\ell)\leq\sum_{\pmb{z}\in\mathbb{Z}^m}\phi^{'''}_k(\pmb{z},\ell)\leq\sum_{\pmb{x}\in\Lambda_c}\tilde{\phi}_k(\pmb{x},\ell),$$
where
$$\tilde{\phi}_k(\pmb{x},\ell)=\begin{cases} S_k(\ell), &\text{if $|\pmb{B}_{\ell}\pmb{x}|^2-2(\pmb{B}_{\ell}\pmb{x})^\mathsf{T}\pmb{e}'_{\ell}\leq bm$;}\cr
                                                          0, &\text{otherwise.}\end{cases}
                                                         $$
Consider now the following lemma:
\begin{lem}
At the $\ell$-th ARQ round, the number of nodes visited by the lattice stack sequential decoder at level $k=1,\cdots,m$, given that $|\pmb{e}'|^2\leq MTL(1+\gamma)$, can be upper bounded by (for any finite $b>0$)
\begin{equation}\label{lemma}
\sum_{\pmb{z}_1^k\in\mathbb{Z}^k}\phi(\pmb{z}_1^k,\ell)\leq S_k(\ell)\leq {{(4\pi)}^{k/2}\over \Gamma(k/2+1)}{[bk+MTL(1+\gamma)]^{k/2}\over \det(\pmb{R}_{kk}^{(\ell)\mathsf{T}}\pmb{R}^{(\ell)}_{kk})^{1/2}}=S'_{k}(\ell),
\end{equation}
where $S_k(\ell)$ is as defined in (\ref{Sk}).
\end{lem}
\begin{proof}
See Appendix II.
\end{proof}

The tail distribution, given the channel is not in outage, can then be upper bounded as follows
\begin{equation}\label{UBb2}
\Pr(\mathcal{N}_m(\ell)\geq \Gamma_{\rm out}|\overline{\mathcal{O}}_{ls}(\ell))\leq \Pr(|\pmb{e}'_{\ell}|^2>MTL(1+\gamma))+\Pr(\mathcal{N}_m(\ell)\geq \Gamma_{\rm out},|\pmb{e}'_{\ell}|^2\leq MTL(1+\gamma)|\overline{\mathcal{O}}_{ls}(\ell)).
\end{equation}

For a given lattice $\Lambda_c$, using Markov inequality, we have
\begin{equation}
\begin{split}
\Pr(\mathcal{N}_m(\ell)\geq \Gamma_{\rm out}|\Lambda_c,\overline{\mathcal{O}}_{ls}(\ell),|\pmb{e}'_{\ell}|^2\leq MTL(1+\gamma))&\leq\Pr(\tilde{\mathcal{N}}_m(\ell)\geq \Gamma_{\rm out}-m|\Lambda_c,\overline{\mathcal{O}}_{ls}(\ell),|\pmb{e}'_{\ell}|^2\leq MTL(1+\gamma))\cr
&\leq{\mathbb{E}_{\pmb{e}'}\{\tilde{\mathcal{N}}_m(\ell)|\Lambda_c,\overline{\mathcal{O}}_{ls}(\ell),|\pmb{e}'_{\ell}|^2\leq MTL(1+\gamma)\}\over \Gamma_{\rm out}-m}, \quad \text{for $\Gamma_{\rm out}>m$,}
\end{split}
\end{equation}
where $\tilde{\mathcal{N}}_m(\ell)$, assuming all-zero codeword was transmitted, is defined as
$$\tilde{\mathcal{N}}_m(\ell)=\sum_{k=1}^{m}\sum_{\pmb{z}_1^k\in\mathbb{Z}^k\backslash\{\pmb{0}\}}\phi(\pmb{z}_1^k,\ell).$$
Using Lemma 2, the conditional average of $\tilde{\mathcal{N}}_m(\ell)$ with respect to the noise can be further upper bounded as
\begin{equation}\label{NUB2}
\begin{split}
\mathbb{E}_{\pmb{e}'}\{\tilde{\mathcal{N}}_m(\ell)|\Lambda_c,\overline{\mathcal{O}}_{ls}(\ell),|\pmb{e}'_{\ell}|^2\leq MTL(1+\gamma)\}&\leq \sum\limits_{k=1}^m S'_k(\ell)\sum\limits_{\pmb{x}\in\Lambda_c^*}\Pr(|\pmb{B}_{\ell}\pmb{x}|^2-2(\pmb{B}_{\ell}\pmb{x})^{\mathsf{T}}\pmb{e}'_{\ell}<bm).
\end{split}
\end{equation}
Therefore,
\begin{equation}\label{UUU}
\Pr(\mathcal{N}_m(\ell)\geq \Gamma_{\rm out}|\Lambda_c,\overline{\mathcal{O}}_{ls}(\ell),|\pmb{e}'_{\ell}|^2\leq MTL(1+\gamma))\leq{\sum\limits_{k=1}^m S'_k(\ell)\over \Gamma_{\rm out}-m}\sum\limits_{\pmb{x}\in\Lambda_c^*}\Pr(|\pmb{B}_{\ell}\pmb{x}|^2-2(\pmb{B}_{\ell}\pmb{x})^{\mathsf{T}}\pmb{e}'_{\ell}<bm).
\end{equation}

We would like now to upper bound the term inside the summation in (\ref{UUU}). The difficulty here stems from the non-Gaussianity of the random vector $\pmb{e}'_\ell$ for any finite $T$. To overcome this problem, consider the following:

Let
$$\tilde{\pmb{e}}_\ell=[\pmb{B}_\ell-\pmb{F}_\ell\pmb{H}_\ell]\pmb{g}_\ell+\pmb{F}_\ell(\pmb{e}_\ell+\pmb{w}_\ell),$$
where $\pmb{g}_\ell\sim\mathcal{N}(0,\sigma^2\pmb{I}_m)$, $\pmb{w}_\ell\sim\mathcal{N}(0,(\sigma^2-1/2)\pmb{I}_m)$ and $\sigma^2\geq1/2$. Following the footsteps of \cite{Label6}, it can be shown that by appropriately constructing a nested LAST code we have that
\begin{equation}\label{NUB3}
\Pr(|\pmb{B}_{\ell}\pmb{x}|^2-2(\pmb{B}_{\ell}\pmb{x})^{\mathsf{T}}\pmb{e}'_{\ell}<bm)\leq\beta_m\Pr(|\pmb{B}_{\ell}\pmb{x}|^2-2(\pmb{B}_{\ell}\pmb{x})^{\mathsf{T}}\tilde{\pmb{e}}_{\ell}<bm),
\end{equation}
where $\tilde{\pmb{e}}_{\ell}\sim\mathcal{N}(0,0.5\pmb{I}_m)$, and $\beta_m$ is a constant independent of $\rho$.
Using Chernoff bound,
\begin{equation}\label{UB5}
\Pr(|\pmb{B}_{\ell}\pmb{x}|^2-2(\pmb{B}_{\ell}\pmb{x})^{\mathsf{T}}\tilde{\pmb{e}}_{\ell}<bm)\leq\begin{cases}e^{-|\pmb{B}_{\ell}\pmb{x}|^2/8}e^{bm/4}, & \text{$|\pmb{B}_{\ell}\pmb{x}|^2>bm$;}\cr 1, &\text{$|\pmb{B}_{\ell}\pmb{x}|^2\leq bm.$} \end{cases}
\end{equation}
By taking the expectation over the ensemble of random lattices (see \cite{Label9}, Theorem 4)
\begin{equation}\label{Mink12}
\mathbb{E}_{\Lambda_c}\left\{\sum\limits_{\pmb{x}\in\Lambda_c^*}\Pr(|\pmb{B}_{\ell}\pmb{x}|^2-2(\pmb{B}_{\ell}\pmb{x})^{\mathsf{T}}\pmb{e}'_{\ell}<bm)\right\}\leq{\beta_m\over V_c}\Bigg\{{\pi^{m/2}(bm)^{m/2}\over\Gamma(m/2+1) \det(\pmb{B}^{\mathsf{T}}_{\ell}\pmb{B}_{\ell})^{1/2}}+ {(8\pi)^{m/2}e^{bm/4}\over  \det(\pmb{B}^{\mathsf{T}}_{\ell}\pmb{B}_{\ell})^{1/2}}\Bigg\}.
\end{equation}

Substituting (\ref{BL}) and (\ref{M}) in (\ref{Mink12}), and by selecting $\Gamma_{\rm out}\geq m+\sum_{k=1}^m S'_k(\ell),$
we get
\begin{equation}\label{UBa2}
\Pr(\mathcal{N}_m\geq \Gamma_{\rm out}|\overline{\mathcal{O}}_{ls}(\ell))\limi{\leq}\rho^{-T\ell[\sum_{i=1}^{M}(1-\alpha_i)^+-{r_1/\ell}]}.
\end{equation}
Now, by setting $\gamma=\zeta\log\rho$, the first term in the RHS of (\ref{UBb2}) can be shown to be upper bounded by $\rho^{-f(r_1/\ell)}$, as long as $\zeta$ is chosen sufficiently large such that $MTL\zeta\geq f(r_1/\ell)$ (see \cite{Label1}, Appendix IV). Averaging the second term in the RHS of (\ref{UBa2}) over the channels in $\overline{\mathcal{O}}_{ls}(\ell)$ set, we obtain ,
\begin{equation}\label{PEs}
\Pr(\mathcal{N}_m(\ell)\geq \Gamma_{\rm out})\limi{\leq}\rho^{-f(r_1/\ell)}+\int_{\overline{\mathcal{O}}_{ls}(\ell)}f_{\pmb{\alpha}}(\pmb{\alpha})\Pr(\mathcal{N}_m\geq \Gamma_{\rm out}|\pmb{\alpha})\;d\pmb{\alpha}\limi{\leq}\rho^{-f(r_1/\ell)},
\end{equation}
where $f_{\pmb{\alpha}}(\pmb{\alpha})$ is the joint probability density function of $\pmb{\alpha}$ which, for all $\pmb{\alpha}\in\overline{\mathcal{O}}_{ls}(\ell)$, is asymptotically given by (see~\cite{Label2})
\begin{equation}\label{pdf}
f_{\pmb{\alpha}}(\pmb{\alpha})\limi{=}\exp\left(-\log(\rho)\sum\limits_{i=1}^{M}(2i-1+N-M)\alpha_i\right).
\end{equation}

Since $\mathcal{N}_j(\ell)<\mathcal{N}_m(\ell)$, for all $1\leq j<m$, then
$$\Pr(\mathcal{N}_j(\ell)\geq \Gamma_{\rm out})\leq \Pr(\mathcal{N}_m(\ell)\geq \Gamma_{\rm out})\limi{\leq} \rho^{-f(r_1/\ell)}.$$

\section*{Appendix II:\\ Proof of Lemma 2}
In the proof of Lemma 1, we have shown that computational complexity at level $k$ when $|\pmb{e}'_{\ell}|^2\leq MTL(1+\gamma)$, can be upper bounded by
$$\sum_{\pmb{z}_1^k\in\mathbb{Z}^k}\phi(\pmb{z}_1^k,\ell)\leq\sum_{\pmb{z}_1^k\in\mathbb{Z}^k}\phi^{''}(\pmb{z}_1^k,\ell)=S_k(\ell),$$
where $\phi^{''}(\pmb{z}_1^k,\ell)$ is the indicator function defined in (\ref{C111}). For brevity, we drop the subscript $\ell$ in the rest of the proof.

Given that $|\pmb{e}|^2\leq R_s^2$, it must follow that $|\pmb{e}_1^k|\leq R_s^2$, where $\pmb{e}_1^k$ is the last $k$ components of $\pmb{e}$. Without loss of generality, we assume that all-zero lattice point was transmitted. Let
 \begin{equation}\label{phi1}
\phi'(\pmb{z}_1^k)=\begin{cases} 1, &\text{if $|{\pmb{e}'}_1^k-\pmb{R}_{kk}\pmb{z}_1^k|^2\leq bk+R_s^2$, $|{\pmb{e}'}_1^k|^2\leq R_s^2$;}\cr
                                                          0, &\text{otherwise.}\end{cases}
                                                          \end{equation}
where $R_s^2=MTL(1+\gamma)$. The total number of integer lattice points that satisfy (\ref{phi1}) is given by
  \begin{equation}\label{L1}
  \mathcal{N}_k\leq\sum\limits_{\pmb{z}_1^k\in\mathbb{Z}^k}\phi'(\pmb{z}_1^k)\leq\sum\limits_{\pmb{z}_1^k\in\mathbb{Z}^k}\phi''(\pmb{z}_1^k) .
  \end{equation}
where
 \begin{equation}\label{phi2}
\phi''(\pmb{z}_1^k)=\begin{cases} 1, &\text{if $|{\pmb{e}'}_1^k-\pmb{R}_{kk}\pmb{z}_1^k|^2\leq bk+R_s^2$, $|{\pmb{e}'}_1^k|^2\leq bk+ R_s^2$;}\cr
                                                          0, &\text{otherwise.}\end{cases}
                                                          \end{equation}
 In general one can show that for any random vectors $\pmb{u}$ and $\pmb{v}$, and $\alpha>0$, it holds$\{|\pmb{u}-\pmb{v}|^2\leq \alpha,|\pmb{u}|^2\leq \alpha\}\subseteq \{|\pmb{u}|^2\leq 4\alpha\}$. Therefore, we can further upper bound (\ref{L1}) as
 \begin{equation}
\mathcal{N}_k\leq \sum\limits_{\pmb{z}_1^k\in\mathbb{Z}^k}\hat{\phi}(\pmb{z}_1^k),
 \end{equation}
 where
 \begin{equation}\label{phi2}
\hat{\phi}(\pmb{z}_1^k)=\begin{cases} 1, &\text{if $|\pmb{R}_{kk}\pmb{z}_1^k|^2\leq 4(bk+R_s^2)$;}\cr
                                                          0, &\text{otherwise.}\end{cases}
                                                          \end{equation}
The summation of $\hat{\phi}(\pmb{z}_1^k)$ over all integer lattice points $\pmb{z}_1^k\in\mathbb{Z}^k$ can then be easily upper bounded by (see [13])
$$\mathcal{N}_k\leq{V(\mathcal{S}_k(2\sqrt{bk+R^2_s}))\over\det(\pmb{R}_{kk}^\mathsf{T}\pmb{R}_{kk})^{1/2}}.$$

\section*{Appendix III:\\ Proof of Theorem 2}
First, the error probability (\ref{Pe}) is lower bounded by the probability of error of the optimal maximum-likelihood or the MMSE-DFE lattice decoder that operates on the whole received signal vector $\pmb{y}=\pmb{y}_L$ knowing the channel matrix $\pmb{H}^c$. Hence, it can be easily shown that (see~\cite{Label8} or \cite{Label2})
\begin{equation}\label{BB0}
P_e(\rho)\geq P(E_{ld}(\pmb{B}_L\pmb{G}))\limi{=}\outage{L}.
\end{equation}
The input to the decoder at round $\ell$, after QR preprocessing of (\ref{estimation1}) is given by
\begin{equation}
\pmb{y}''_{\ell}=\pmb{Q}_{\ell}^\mathsf{T}\pmb{y}_{\ell}=\pmb{R}_{\ell}\pmb{z}+\pmb{e}''_{\ell},
\end{equation}
where $\pmb{e}''_{\ell}=\pmb{Q}_{\ell}^\mathsf{T}\pmb{e}'_{\ell}$.

Consider the long-term static channel. For the probability of error we bound each term in (\ref{Pe}). At round $L$, due to lattice symmetry, we assume that the all zero codeword, i.e., $\pmb{0}$, was transmitted. At high SNR, we have shown in section III that, for a given lattice $\Lambda_c$, the error probability of the lattice stack sequential decoder can be upper bounded as
\begin{equation}\label{bounds1}
P(E_{sd}(\pmb{B}_L,b)|\Lambda_c)\limi{\leq} P(E_{ld}(\tilde{\pmb{B}}_L)|\Lambda_c).
\end{equation}
where $\tilde{\pmb{B}}_L=(1-\epsilon)\pmb{B}_L$, and $0\leq\epsilon\leq 1$. Using the union bound, we can further upper bound (\ref{bounds1}) by
\begin{equation}\label{UB3}
P(E_{sd}(\pmb{B}_L,b)|\Lambda_c)\limi{\leq} \sum\limits_{\pmb{x}\in\Lambda_c^{*}}\Pr(2(\tilde{\pmb{B}}_L\pmb{x})^{\mathsf{T}}\pmb{e}'>|\tilde{\pmb{B}}_L\pmb{x}|^2).
\end{equation}

As shown in the proof of Theorem 1 (see Appendix I), by appropriately constructing a nested LAST code we have that
\begin{equation}\label{UB4}
P(E_{sd}(\pmb{B}_L,b)|\Lambda_c)\leq \beta_m\sum\limits_{\pmb{x}\in\Lambda_c^{*}}\Pr(2(\tilde{\pmb{B}}_L\pmb{x})^{\mathsf{T}}\tilde{\pmb{e}}>|\tilde{\pmb{B}}_L\pmb{x}|^2),
\end{equation}
where $\tilde{\pmb{e}}\sim\mathcal{N}(0,1/2\pmb{I}_m)$, and $\beta_m$ is a constant independent of $\rho$. Using Chernoff bound,
\begin{equation}\label{UB51}
\Pr(2(\pmb{B}'_L\pmb{x})^{\mathsf{T}}\tilde{\pmb{e}}>|\tilde{\pmb{B}}_L\pmb{x}|^2)\leq e^{-|\tilde{\pmb{B}}_L\pmb{x}|^2/8}.
\end{equation}
By taking the expectation over the ensemble average of random lattices (see \cite{Label9}, Theorem 4)
\begin{equation}\label{Mink1}
\begin{split}
&P(E_{sd}(\pmb{B}_L,b))=\mathbb{E}_{\Lambda_c}\{P(E_{sd}(\pmb{B}_L,b)|\Lambda_c)\}\cr
&\leq{\beta_m\over V_c}\int\limits_{|\pmb{B}'_L\pmb{x}|^2>0}e^{-|\pmb{B}'_L\pmb{x}|^2/8}\;d\pmb{x}\leq{(8\pi)^{m/2}\over  V_c(1-\epsilon)^{m/2}\det(\pmb{B}_{L}^{\mathsf{T}}\pmb{B}_{L})^{1/2}}.
\end{split}
\end{equation}
Substituting (\ref{BL}) (with $\ell=L$) and (\ref{M}) in (\ref{Mink1}), the average error probability at round $L$, given the channel is not in outage can be rewritten as
\begin{equation}
P(E_{sd}(\pmb{B}_L,b)|\overline{\mathcal{O}}_{ls}(L))\limi{\leq} \mathcal{K}(m,\epsilon)\rho^{-TL[\sum_{j=1}^{M}(1-\alpha_j)^{+}-{r_1\over L}]},
\end{equation}
where $\mathcal{K}(m,\epsilon)$ is a constant independent of $\rho$. Using the results in \cite{Label8}, the average probability of error (average over the channel and lattice ensemble) is asymptotically upper bounded by
\begin{equation}\label{BB1}
\begin{split}
 P(E_{sd}(\pmb{B}_L,b))&\leq\Pr(\mathcal{O}_{ls}(L))+ P(E_{sd}(\pmb{B}_L,b),\overline{\mathcal{O}}_{ls}(L))\\
&\limi{\leq}\rho^{-f(r_1/L)} + \int_{\overline{\mathcal{O}}_{ls}(L)}f_{\pmb{\alpha}}(\pmb{\alpha})P(E_{sd}(\pmb{B}_L,b)|\overline{\mathcal{O}}_{ls}(L))\;d\pmb{\alpha}\limi{\leq}\rho^{-f(r_1/L)},
\end{split}
\end{equation}
under the condition that $LT\geq M+N-1$, where $f_{\pmb{\alpha}}(\pmb{\alpha})$ is given in (\ref{pdf}).

The proof for the short-term static channel follows the same lines with the exception that
$$\det(\pmb{B}_{L}^{\mathsf{T}}\pmb{B}_{L})=\sum\limits_{j=1}^{L}\det\left(\pmb{I}+{\rho\over M}(\pmb{H}_j^c)^\mathsf{H}\pmb{H}_j^c\right)^{2T}.$$
In this case, one can easily verify that
\begin{equation}\label{BB2}
 P(E_{sd}(\pmb{B}_L,b))\limi{\leq}\rho^{-Lf(r_1/L)},
\end{equation}
under the condition that $T\geq M+N-1$.

Thus, the error probability in (\ref{Pe}) can be asymptotically upper bounded as
\begin{equation*}
P_e(\rho)\limi{\leq}\begin{cases}\rho^{-f(r_1/L)},& \textit{for long-term static channel}; \cr
\rho^{-Lf(r_1/L)},& \textit{for short-term static channel}.
\end{cases}
\end{equation*}

Next, we need to show that $r_e\limi{=}r_1$. As discussed previously, when the channel is not in outage, a retransmission is requested by the time-out lattice sequential decoder at round $\ell$, whenever the number of computations $\mathcal{N}_j(\ell)$ at any level $1\leq j\leq m$ exceeds $\Gamma_{\rm out}$. Therefore, the probability of a retransmission at round $\ell$ when the channel is not in outage, can be expressed as
\begin{eqnarray*}
\label{1}
\Pr(\overline{\mathcal{A}}_{\ell},\overline{\mathcal{O}}_{ls}(\rho,\ell))=\Pr\biggl(\bigcup_{j=1}^m\{\mathcal{N}_j(\ell)\geq \Gamma_{\rm out}\}\biggr).
\end{eqnarray*}
And since $\mathcal{N}_m(\ell)> \mathcal{N}_j(\ell)$ $\forall 1\leq j\leq m$, using the union bound we can upper bound the above probability by
\begin{equation*}
\Pr(\overline{\mathcal{A}}_{\ell},\overline{\mathcal{O}}_{ls}(\rho,\ell))\leq m\Pr(\mathcal{N}_m(\ell)\geq \Gamma_{\rm out}).
\end{equation*}

Therefore, we have
\begin{eqnarray}
p(\ell)&\leq&{\rm Pr}(\overline{\mathcal{A}}_{\ell})\cr
          &=&{\rm Pr}(\mathcal{O}_{ls}(\rho,\ell))+{\rm Pr}(\overline{\mathcal{A}}_{\ell}, \overline{\mathcal{O}}_{ls}(\rho,\ell)).
\end{eqnarray}
The behavior of the first term at high SNR is $\rho^{-f(r_1/\ell)}$. Using Lemma 1, we get
\begin{equation*}
{\rm Pr}(\overline{\mathcal{A}}_{\ell}, \overline{\mathcal{O}}_{ls}(\rho,\ell))\limi{\leq}\rho^{-d_{\ell}}.
\end{equation*}
where $d_{\ell}$ is as defined in Theorem 1.

Following \cite{Label8}, for any $T\geq 1$ we find that $d_{\ell}>0$ for all $\ell$ and $r_1<M$. Moreover, if $T\ell\geq M+N-1$ then $d_{\ell}=f(r_1/\ell)$, which is the maximum possible SNR exponent for codes with multiplexing gain $r_1/\ell$ and block length $T\ell$.

Therefore, we have that
$$p(\ell)\limi{\leq}\rho^{-\min\{f(r_1/\ell),d_{\ell}\}}.$$
Using (\ref{eta}), we obtain
$${R_1\over1+\sum_{\ell=1}^L\rho^{-\min\{f(r_1/\ell),d_{\ell}\}}}\limi{\leq}\eta\leq R_1.$$
This implies that $r_e\limi{=}r_1$, and therefore for the long-term static channel, $d_{ls}^*(r_e,L)$ is achievable. The proofs for the short-term static channel follow exactly the same arguments, and one can show that $d_{ss}^*(r_e,L)$ is achievable under time-out lattice sequential decoding.



\begin{figure}[ht!]
\begin{center}
\includegraphics[width=5in]{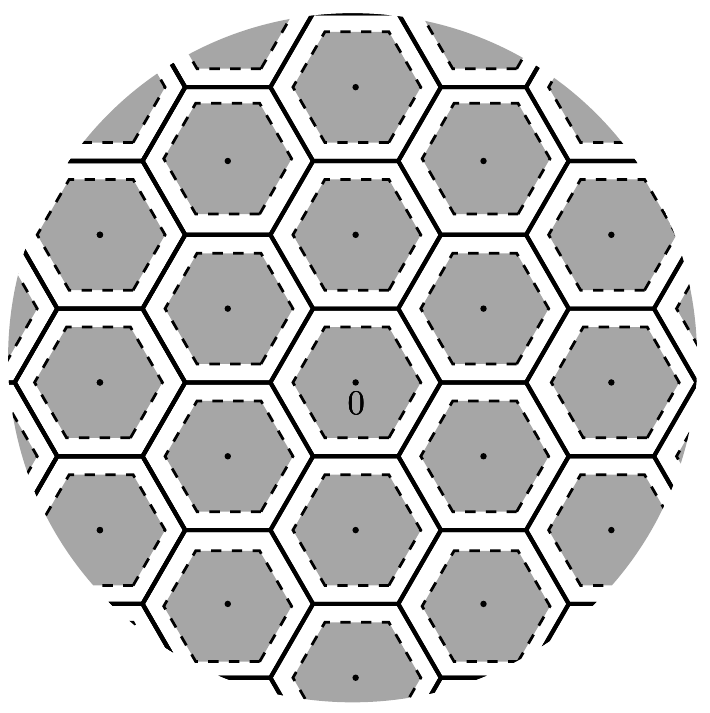}
\caption{The division of the Voronoi cell of the lattice generated by the channel-code matrix $\pmb{B}_\ell\pmb{G}$ into two distinct regions --- the shaded region $\mathcal{R}_{\pmb{u}}(\tilde{\pmb{B}}_\ell\pmb{G})$, and the white region $\mathcal{V}_{\pmb{u}}(\pmb{B}_\ell\pmb{G})\backslash\mathcal{R}_{\pmb{u}}(\tilde{\pmb{B}}_\ell\pmb{G})$. The two dimensional hexagonal lattice is shown for illustration purposes.}
\end{center}
\end{figure}

\begin{figure}[ht!]
\centering

  \subfigure[The event of sending a NACK when the channel is not in outage.]{
   \includegraphics[width =2in] {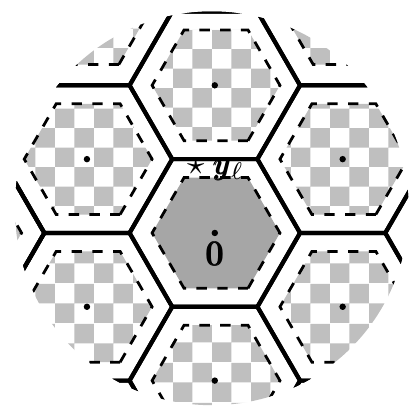}
 }

\subfigure[The event of sending an ACK with correct decoding.]{
   \includegraphics[width =2in] {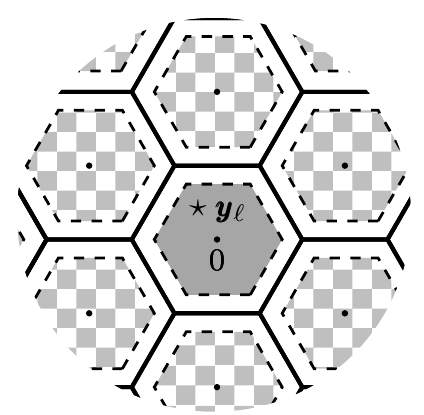}
 }
\quad
 \subfigure[The event of sending an ACK with decoding failure but not detected.]{
   \includegraphics[width=2in] {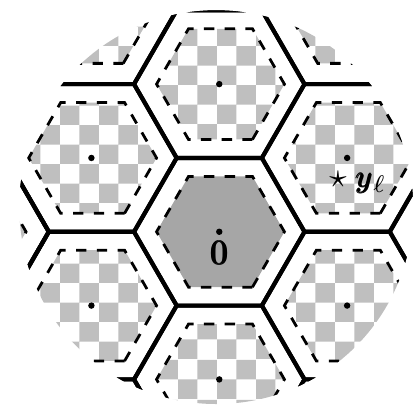}
 }

\caption{The events of retransmission, correct decoding, and undetected error that occur in the time-out algorithm (assuming $\pmb{0}$ was transmitted). The correct decoding region is represented in dark color. The chessboard shaded regions represent the undetected error events $(E_{\ell},\mathcal{A}_{\ell})$. The white region represents the undetected error event.}
\end{figure}

\begin{figure}[ht!]
\begin{center}
\includegraphics[width=4in]{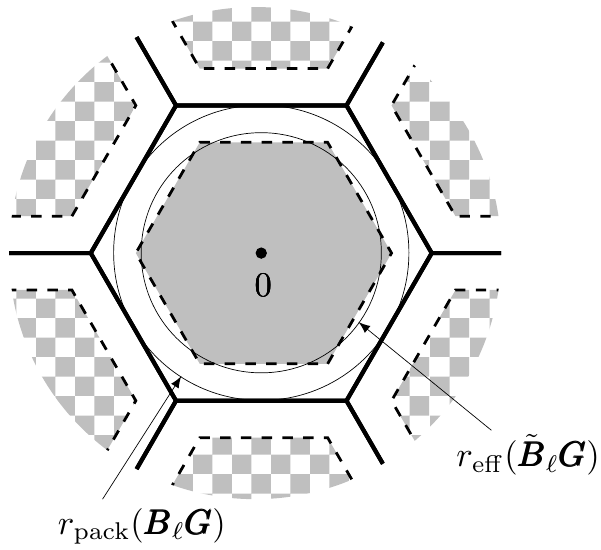}
\caption{A geometric approach used to over bound the undetected error probability under the time-out algorithm.}
\end{center}
\end{figure}

\begin{figure}[ht!]
\begin{center}
\includegraphics[width=4in]{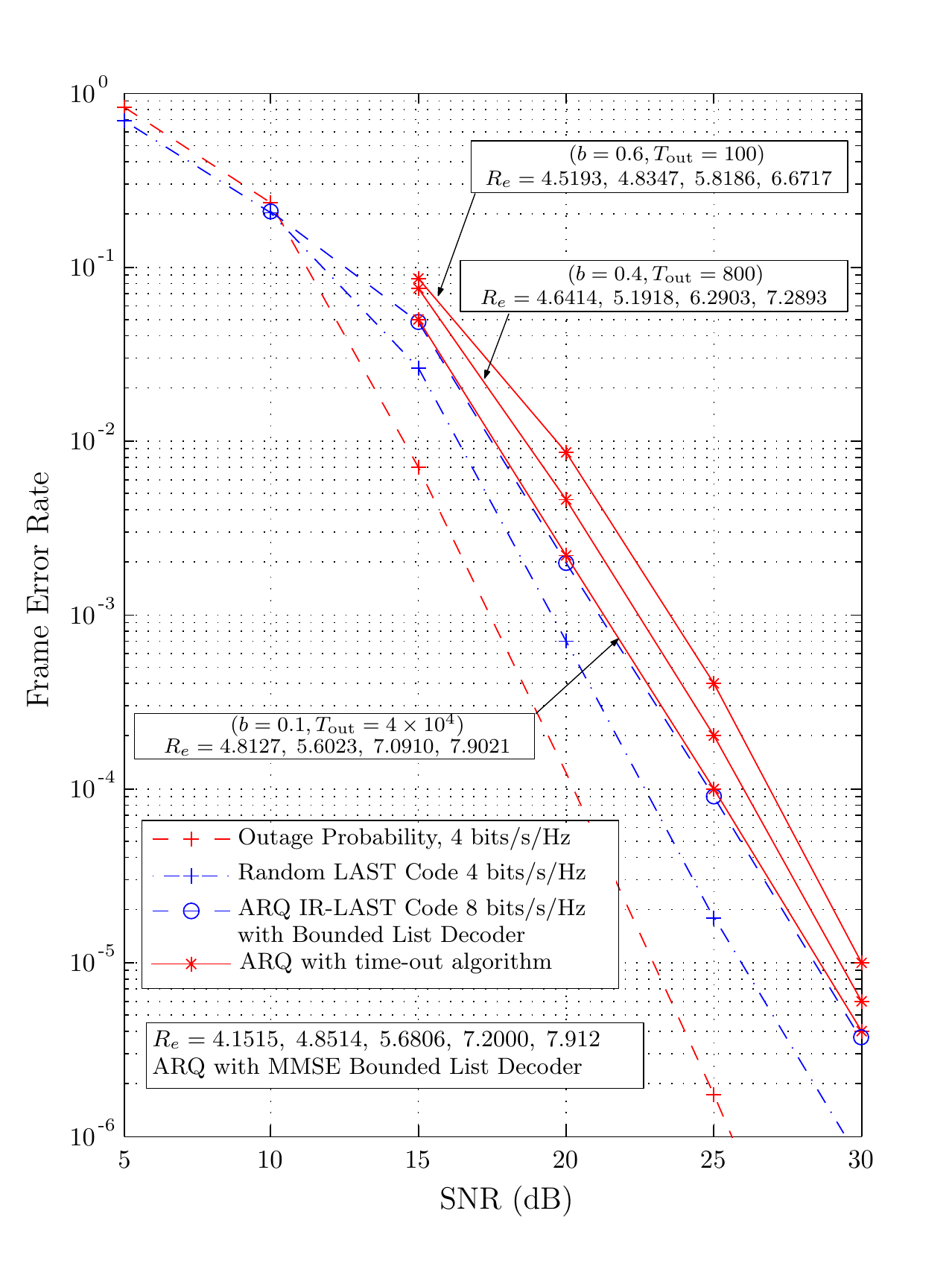}
\caption{The optimal tradeoff achieved by the time-out algorithm lattice stack sequential decoder for several values of $b$.}
\label{default}
\end{center}
\end{figure}
\begin{figure}[!ht]
\begin{center}
\includegraphics[width=5in]{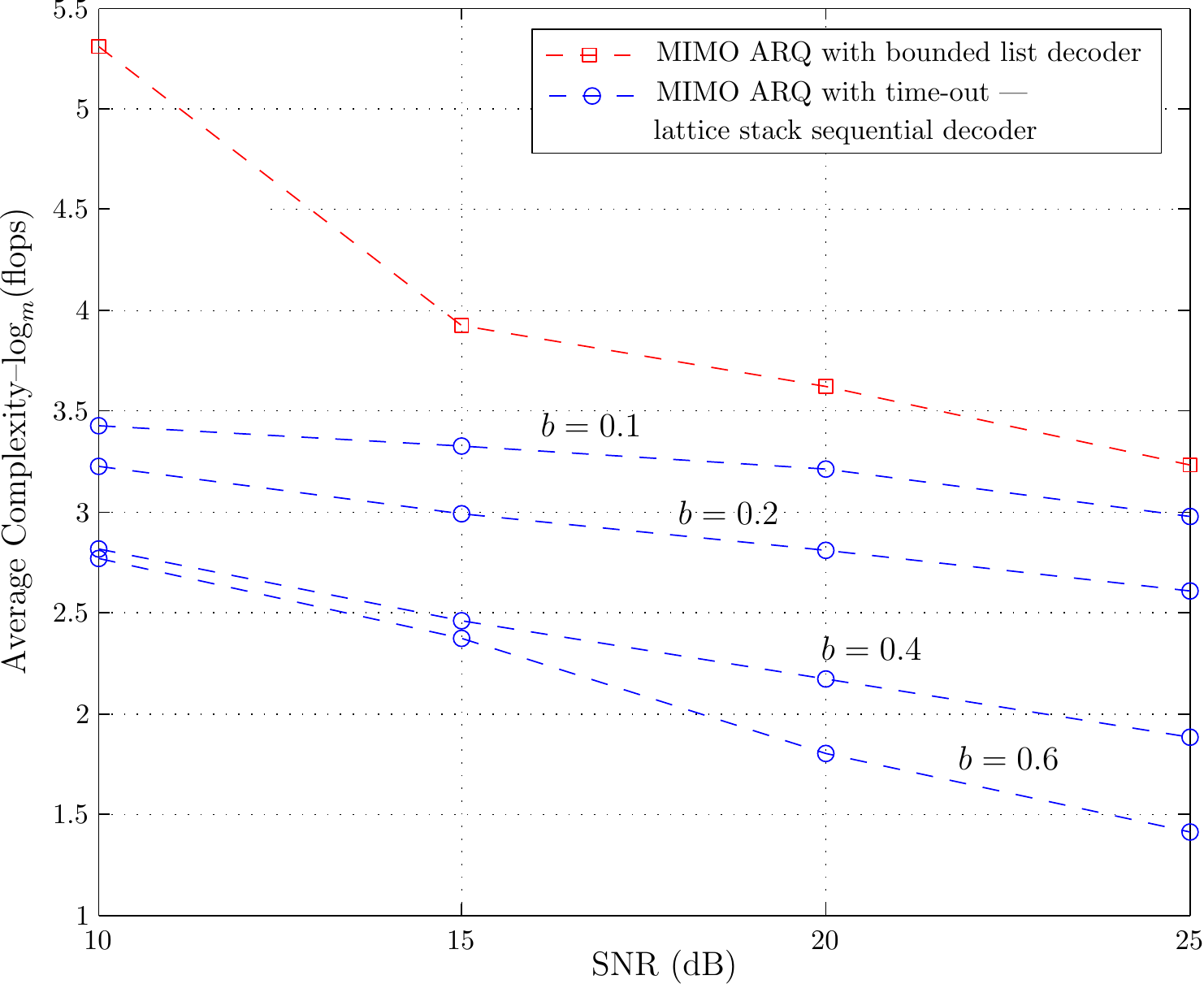}
\caption{Comparison of average computational complexity of the MMSE-DFE list lattice decoder and the time-out lattice stack sequential decoder for several values of $b$ using their corresponding optimal values of $\Gamma_{\rm out}$ (see Fig.~4).}
\end{center}
\end{figure}

\end{document}